\documentclass[letterpaper, 10 pt, conference]{ieeeconf}  % Comment this line out if you need a4paper

\IEEEoverridecommandlockouts                              % This command is only needed if 
\usepackage{graphicx} % for pdf, bitmapped graphics files
\usepackage{amsmath, dsfont} % assumes amsmath package installed
\usepackage{amssymb}  % assumes amsmath package installed
\usepackage{cleveref}
\usepackage{subfiles}
\newtheorem{theorem}{Theorem}

\newtheorem{lemma}{Lemma}
\DeclareMathOperator*{\argmin}{argmin} % thin space, limits underneath in displays

\usepackage[linesnumbered,ruled,vlined]{algorithm2e}
\SetKwInput{KwInput}{Input}                % Set the Input
\SetKwInput{KwOutput}{Output}              % set the Output

\usepackage{color}

\title{\LARGE \bf
Approximate Multiagent Reinforcement Learning for On-Demand Urban Mobility Problem on a Large Map (extended version)}

\author{Daniel Garces$^{1}$, Sushmita Bhattacharya$^{1}$, Dimitri Bertsekas$^{2}$, Stephanie Gil$^{1}$
\thanks{$^{1}$Daniel Garces, Sushmita Bhattacharya, and Stephanie Gil are with the REACT lab, Harvard University, Boston, MA, USA
(e-mail: {\tt\small \{dgarces, sushmita\_bhattacharya,sgil\}@g.harvard.edu})} %{\tt\small sushmita\_bhattacharya@g.harvard.edu}, {\tt\small sgil@g.harvard.edu}). }%
\thanks{$^{2}$Dimitri Bertsekas is with the Department of Electrical Engineering and Computer Science, Arizona State University, AZ, USA (e-mail: {\tt\small dimitrib@mit.edu}).}
\thanks{This work was supported by  
ONR YIP grant \#N00014-21-1-2714, NSF grant \#2114733, AFOSR grant \#FA9550-22-1-0223,
Amazon Research Award, and in part by Apple Scholars in AI/ML Ph.D. Fellowship Program. }
}

\begin{document}

\def\paperlongversion{0}

\maketitle
\thispagestyle{empty}
\pagestyle{empty}

%%%%%%%%%%%%%%%%%%%%%%%%%%%%%%%%%%%%%%%%%%%%%%%%%%%%%%%%%%%%%%%%%%%%%%%%%%%%%%%%
\begin{abstract}
In this paper, we focus on the autonomous multiagent taxi routing problem for a large urban environment where the location and number of future ride requests are unknown a-priori, but can be estimated by an empirical distribution. Recent theory has shown that a rollout algorithm with a stable base policy produces a near-optimal stable policy. In the routing setting, a policy is stable if its execution keeps the number of outstanding requests uniformly bounded over time. Although, rollout-based approaches are well-suited for learning cooperative multiagent policies with considerations for future demand, applying such methods to a large urban environment can be computationally expensive due to the large number of taxis required for stability. In this paper, we aim to address the computational bottleneck of multiagent rollout by proposing an approximate multiagent rollout-based two phase algorithm that reduces computational costs, while still achieving a stable near-optimal policy. Our approach partitions the graph into sectors based on the predicted demand and the maximum number of taxis that can run sequentially given the user's computational resources. The algorithm then applies instantaneous assignment (IA) for re-balancing taxis across sectors and a sector-wide multiagent rollout algorithm that is executed in parallel for each sector. We provide two main theoretical results: 1) characterize the number of taxis $m$ that is sufficient for IA to be stable; 2) derive a necessary condition on $m$ to maintain stability for IA as time goes to infinity. Our numerical results show that our approach achieves stability for an $m$ that satisfies the theoretical conditions. We also empirically demonstrate that our proposed two phase algorithm has equivalent performance to the one-at-a-time rollout over the entire map, but with significantly lower runtimes.
\end{abstract}

%%%%%%%%%%%%%%%%%%%%%%%%%%%%%%%%%%%%%%%%%%%%%%%%%%%%%%%%%%%%%%%%%%%%%%%%%%%%%%%%
\section{Introduction}
Self-driving taxis are currently operating in multiple cities, including Austin, Phoenix, and San Francisco \cite{Bidarian2023}, with possibilities of being deployed to more cities in the near future \cite{Muller2023}. This widespread deployment of autonomous taxis creates new opportunities for improved on-demand mobility through coordinated routing and planning, and poses interesting new practical and theoretical problems for the field of robotics. For instance, the ability of autonomous taxis to communicate with each other and with a centralized server allows for the orchestration of fleet-wide coordinated plans that result in more requests being serviced \cite{Kondor2022}.

Coordination plans have been studied in the literature in the form of the Dynamic Vehicle Routing (DVR) problem \cite{BERBEGLIA20108} with stochastic demand, where the location and number of future requests is unknown a-priori. However, due to the size of the problem and the complexity associated with the stochasticity of the demand, there are still many research opportunities related to the design of better and faster algorithms to learn cooperative plans that take into account future requests and maximally use taxi fleets. Approaches in the literature have mainly focused on immediate demand \cite{Duan2014, Bertsimas2019OnlineVR} and sector level routing \cite{alonso2017predictive, LOWALEKAR201871, iglesias2018data, gammelli2022graph, enders2023hybrid}, abstracting away either the stochasticity of the demand, or the complexity associated with ''fine-grained'' street/intersection level decisions. Other works, including our previous work \cite{Garces_2023}, have considered using reinforcement learning methods, particularly rollout-based approaches \cite{Bertsekas2021PI, bertsekas2020rollout, Bertsekas2022AlphaZero}, to tackle fine-grained routing decisions. These rollout-based methods as defined in \cite{bertsekas2020rollout} are comprised of three major components: 1) a one-step lookahead cost minimization where the immediate future is simulated using Monte-Carlo approximation for all potential controls, 2) a future cost approximation for each potential control based on a truncated application of a simple to compute policy known as the base policy for a finite time horizon, and 3) a terminal cost approximation that compensates for the truncated application of the base policy. Recent theory \cite{bertsekas2020rollout, Bertsekas2022AlphaZero} shows that rollout's one-step lookahead cost minimization acts as a Newton step, and hence provides super linear convergence to the optimal policy. In particular, as long as the base policy is close to the optimal policy with a reasonable competitive factor \cite{gerkey2004formal} and it is stable, then rollout-based approaches learn a stable near-optimal policy. This theoretical result makes rollout-based algorithms very well-suited for tackling the fine-grained routing problem. In the routing setting, a  policy is said to be stable if its execution results in the number of outstanding requests being uniformly bounded over time. Applying these rollout methods to a large urban environment, however, poses a unique set of challenges that we aim to address in this paper.

A major challenge of dealing with a city-scale environment is the large volume of requests that enters the system, which then requires a large number of taxis to guarantee stability. This large number of taxis makes the application of a multiagent (one-at-a-time) rollout scheme, as proposed in our previous work \cite{Garces_2023}, computationally prohibitive. In this paper, we address this computational bottleneck by proposing an approximation to the one-at-a-time rollout algorithm that keeps computational costs below user-defined constraints, while still maintaining stability and the Newton-step property of rollout. 
Our proposed method reduces the computational cost of executing one-at-a-time rollout with a large number of taxis by partitioning the map into disjoint sectors based on expected demand and the maximum number of taxis that can be run sequentially given the user's indicated computational resources. Our method then executes a two-phase algorithm composed of a high level planner and multiple low level planners that are run in parallel. The high level planner routes taxis across sectors based on the current and estimated future demand, while the low level planners route taxis within each sector by employing one-at-a-time rollout with instantaneous assignment with reassignment (IA-RA) as the base policy. We choose IA-RA as the base policy since it is 2-competitive\footnote{An $\alpha-$competitive policy never produces a cost greater than $\alpha$ times the optimal cost for any input~\cite{gerkey2004formal}}~\cite{gerkey2004formal}, which facilitates the super-linear convergence of one-at-a-time rollout to the optimal policy~\cite{Bertsekas2022AlphaZero}. We provide theoretical results for a sufficient condition on the total number of taxis $m$ that will guarantee IA-RA to be stable. Compared to previous work \cite{Zhang2016Queue, Vazifeh2018, Treleaven2013, spieser2014}, our analysis uses the full stochasticity of the system and assumes that the pickups and dropoffs are jointly distributed. In addition, for the case where pickups and dropoffs can be assumed independent, we also provide a necessary condition on $m$ for asymptotic stability of IA-RA as time goes to infinity, building on the results proposed in \cite{Treleaven2013}. We empirically demonstrate that our approach results in a significantly lower computational cost and comparable performance as one-at-a-time rollout over the entire map, and we verify that stability is achieved for fleet sizes lying within the range given by our theoretical results.

\section{Problem formulation}
In this section, we present the formulation of a large scale multiagent taxicab routing and pickup problem as a discrete time, finite horizon, stochastic Dynamic Programming problem that plans over a city-scaled street network. In the following subsections, we provide definitions for our environment, requests, state and control spaces, the concept of stability, and the challenges associated with the large scale.

\subsection{Environment}
We assume that autonomous taxis are deployed in an urban environment with a known fixed street topology. The environment is hence represented as a directed graph $G=(V, E)$, where $V=\{ 1, \dots, n\}$ corresponds to the set of street intersections in the map numbered $1$ through $n$, while $E \subseteq \{ (i,j) | i, j \in V\}$ corresponds to the set of directed streets that connect intersections $i$ and $j$ (see Fig. \ref{fig:env_map}). The set of neighboring intersections to intersection $i$ is denoted as $\mathcal{N}_i = \{ j | j\in V, (i,j) \in E\}$. 
We also assume that the environment can be divided into $K$ sectors, where each sector $s_k \subseteq V$, and $s_k \neq \emptyset, \forall k = \{1,\ldots, K\},$ such that $V = \bigcup_{k=1}^K s_k$ and $s_k \cap s_h = \emptyset, \forall h \neq k$. 
%We denote the set of all sectors in the map as the set $S$, where $S$ is then a set of sets.

\begin{figure}[ht]
        \centering        \includegraphics[width=0.35\textwidth]{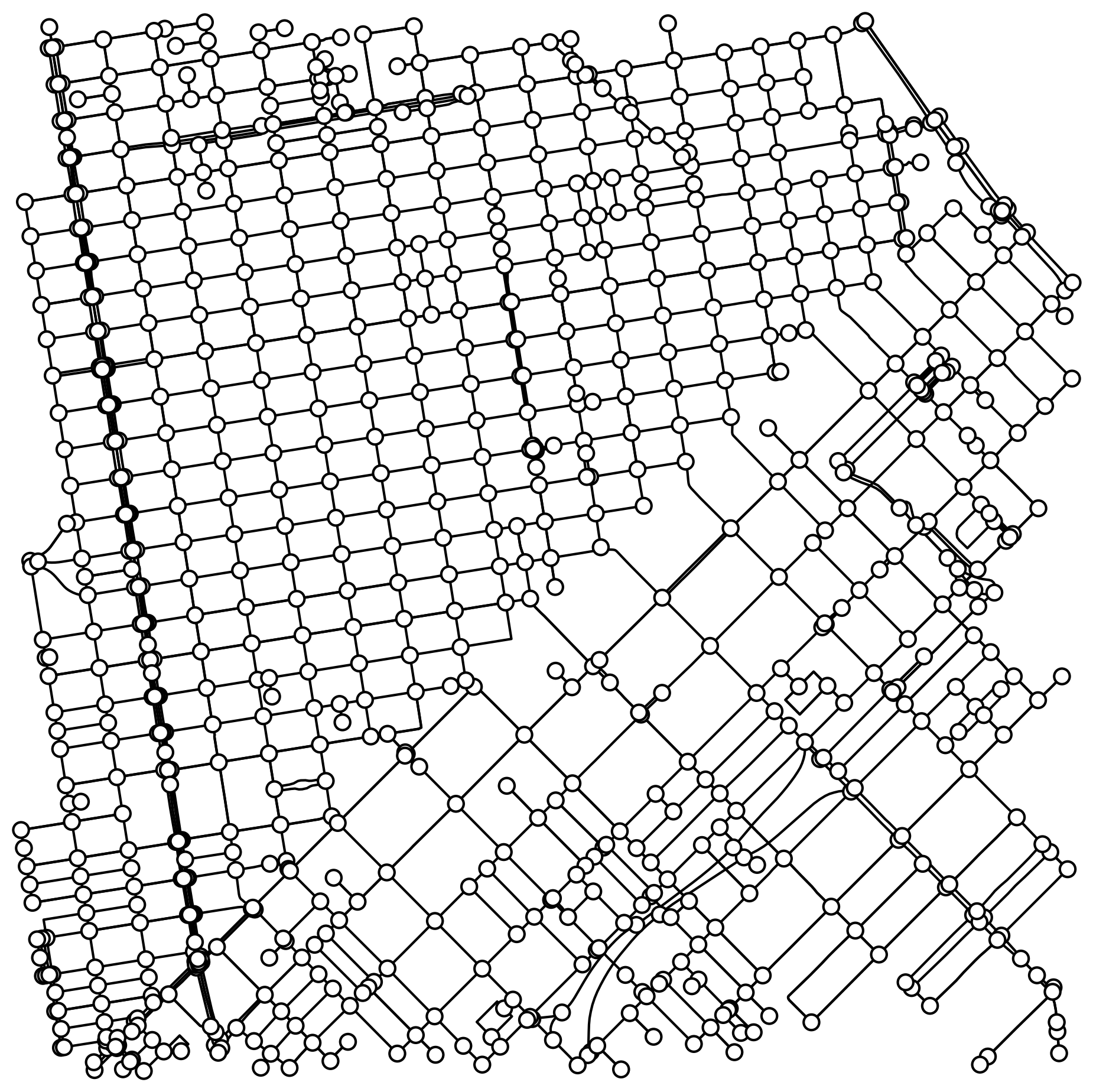}
        \caption{\small{Street network used in numerical experiments}}
        \label{fig:env_map}
        \vspace{-10pt}
\end{figure}
\subsection{Requests}
We define a ride request $r$ as a tuple $r=\left< \rho_r, \delta_r, t_r, \phi_r \right>$, where $\rho_r \in V$ and $\delta_r \in V$ correspond to the nearest intersection to the request's desired pickup and drop-off locations, respectively; $t_r$ corresponds to the time at which the request was placed into the system; and $\phi_r \in \{ 0, 1\}$ is an indicator, such that $\phi_r = 1$ if the request has been picked up by a vehicle, $\phi_r = 0$ otherwise. We model the number of requests that enter the system at time $t$ as a random variable $\eta_t$, which has the same distribution as random variable $\eta$ with an unknown underlying distribution $p_\eta$. We assume $p_\eta$ is fixed for the entire length of the time horizon $T$, and its estimated probability distribution, denoted $\tilde{p}_\eta$, can be estimated from historical trip data. We denote the set of ride requests that enter the system at time $t$ as $\mathbf{r_t}$. Here the cardinality of the set of new requests at time $t$ is $|\mathbf{r_t}|=\eta_t$. We model the pickup intersection for an arbitrary request $r$ as the random variable $\rho_{r}$. Similarly, we model the drop-off intersection for request $r$ as the random variable $\delta_{r}$. We assume that requests are independent and identically distributed (i.i.d), and hence we drop the subscripts when talking about their distributions. Random variables $\rho$ and $\delta$ are jointly distributed and have unknown underlying probability distributions $p_{\rho}$ and $p_{\delta | \rho}$, respectively. We assume these distributions do not change over the entire length of the time horizon $T$.  We denote the marginal distribution of $\delta$ as $p_{\delta}$. We also denote the estimated categorical distributions for pickup locations, conditional dropoff locations, and the marginal dropoff locations as $\tilde{p}_{\rho}$, $\tilde{p}_{\delta | \rho}$, and $\tilde{p}_{\delta}$, respectively. These categorical distributions are estimated using historical trip data.  We define $\mathbf{\overline{r}_{t}} = \{ r | r \in \mathbf{r_{t'}}, \phi_r=0, 1 \leq t' \leq t \}$ as the set of outstanding ride requests that have not yet been picked up by any taxi at time $t$.

\subsection{State and control space}
We assume there is a total of $m$ taxis and all taxis can perfectly observe all requests, and other taxis' locations and occupancy status. We assume that all of the taxis remain inside the predefined street network $G$, and they are able to traverse any edge in $G$ in a single time step. We represent the state of the system at time $t$ as a tuple $x_t = \left< \vec{\nu_t}, \vec{\tau_t}, \mathbf{\overline{r}_{t}} \right>$. We define $\vec{\nu_t}=[\nu_t^1, \dots, \nu_t^m]$ as the list of locations for all $m$ taxis at time $t$, where $\nu_t^\ell \in V$ corresponds to the index of the closest intersection to the geographical position of taxi $\ell$. We define $\vec{\tau_t} = [\tau_t^1, \dots, \tau_t^m]$ as the list of time remaining in the current trip for all $m$ taxis. If taxi $\ell$ is available, then it has not picked up a request and hence $\tau_t^\ell = 0$, otherwise $\tau_t^\ell \in \mathds{N}^+$. The initial location of an arbitrary taxi $\ell$ at time $t=0$ is given by random variable $\xi_\ell$. All $\xi_\ell$ for $\ell = 1, \dots, m$ are assumed to be independent and identically distributed with known underlying distribution $p_\xi$.

We denote the control space for taxi $\ell$ at time $t$ as $\mathbf{U}_t^\ell(x_t)$. If the taxi is available (i.e. $\tau_t^\ell = 0$), then $\mathbf{U}_t^\ell(x_t) = \{ \mathcal{N}_{\nu_t^\ell}, \nu_t^\ell, \psi_r\}$, where %$\mathcal{N}_{\nu_t^\ell}$ corresponds to the set of neighboring intersections to the current location of taxi $\ell$, $\nu_t^\ell$ corresponds to the current location (intersection) of taxi $\ell$, and 
$\psi_r$ corresponds to a special pickup control that becomes available if there is a request $r \in \mathbf{\overline{r}_{t}}$ with pickup at the location of taxi $\ell$ (i.e. $\rho_r = \nu_t^\ell$). If the taxi is currently servicing a request $r$ (i.e. $\tau_t^\ell > 0$), then $\mathbf{U}_t^\ell(x_t) = \{ \zeta \}$, where $\zeta$ corresponds to the next hop in shortest path between taxi $\ell$'s current location $\nu_t^\ell$ and the destination of the request $\delta_r$. The controls available to all $m$ taxis at time $t$,  $\mathbf{U}_t (x_t)$, is expressed as the Cartesian product of local control sets for each taxi, such that $\mathbf{U}_t (x_t) = \mathbf{U}_t^1(x_t) \times \dots \times \mathbf{U}_t^m(x_t)$.

\subsection{Stability of a policy}
\label{sec:stability_definition}
 We define a policy $\pi = \{ \mu_1, \dots \mu_T\}$ as a set of functions that maps state $x_t$ into control $u_t = \mu_t(x_t) \in \mathbf{U}_t (x_t)$. Using a similar formulation as in \cite{spieser2014}, we define the total distance to be traveled in service of a request $r_q$ with index $q$ given a policy $\pi$ as $W_{r_q,\pi} = d(l_{r_q,\pi}, \rho_{r_q}) + d(\rho_{r_q}, \delta_{r_q})$, where $l_{r_q, \pi}$ is the location of a taxi assigned to request $r_q$ based on policy $\pi$, and $d:V \times V \to \mathbb{N}^+$ is a function that gives the length of the shortest path between two locations. We define the total distance to be traveled in service of all the requests that enter the system for the entire time horizon $T$ as $Z_{\pi, T}=\sum_{t=1}^T \sum_{q=R_{t-1}+1}^{R_{t}} W_{r_q,\pi}$, where the random variable $R_{t} = \sum_{t'=1}^t \eta_{t'}$ represents the total number of requests that have entered the system until time $t$. It is important to note that $R_0 = 0$. We define the total distance that can be covered by a fleet of $m$ taxis as $m \cdot T$ since each taxi can travel unit distance at each time step. Assuming that we have at least as many available taxis at each time step as incoming requests, a given policy $\pi$ is said to be stable if, for a fixed fleet size of $m$ taxis, the expected number of outstanding requests is uniformly bounded. Hence, a policy $\pi$ is stable as long as the distance to be traveled in service of all the requests that enter the system according to policy $\pi$ is less than or equal to the total distance that can be covered by a fleet of $m$ taxis. In other words, for a policy $\pi$ to be stable (following a similar argument as in~\cite{spieser2014}), the expected total distance for servicing all requests should be upper bounded by the distance covered by taxis, i.e., $E[Z_{\pi,T}] \leq m\cdot T.$

\subsection{Challenges of a large scale multi-agent problem}\label{subsec:scalability_challenge}
We are interested in learning a cooperative pickup and routing policy on a city-scale map that minimizes the total wait time for all requests over a finite horizon of length $T$. We denote the state transition function as $f$, such that $x_{t+1} = f\left(x_t, u_t, \eta, \rho, \delta \right)$, where $x_{t+1}$ is the resulting state after control $u_t \in \mathbf{U}_t (x_t)$ has been applied from state $x_t$. We define the stage cost $g_t\left(x_t, u_t, \eta, \rho, \delta \right) = |\mathbf{\overline{r}_{t}}|$ as the number of outstanding requests at time $t$. We denote the cost of executing policy $\pi$ from initial state $x_1$ as $J_{\pi}(x_1) = E \left[g_T(x_T) + \sum_{t=1}^{T-1} g_t\left(x_t, \mu_t(x_t), \eta, \rho, \delta \right) \right]$, where $g_T(x_T) = |\overline{\mathbf{r_{T}}}|$ is the terminal cost. Since the control space for the problem grows exponentially with the number of taxis, obtaining an optimal policy through the Bellman equations is intractable. For this reason, we consider policy improvement schemes, such as one-at-a-time rollout \cite{Bertsekas2021PI, bertsekas2020rollout}%,BERTSEKAS2020Multiagent}
, which solve several smaller lookahead optimizations to obtain a lower cost policy that improves upon a base policy and has a control space that scales linearly with the number of taxis instead of exponentially. We define base policy $\pi = \{ {\mu}_1, \dots {\mu}_T \}$ as an easy to compute heuristic that is given. One-at-a-time rollout finds an approximate policy $\tilde{\pi} = \{ \tilde{\mu}_1, \dots \tilde{\mu}_T \}$, where $\tilde{\mu}_t(x_t) = (\tilde{\mu}^1_t(x_t), \ldots, \tilde{\mu}^m_t(x_t))$, $t=[1,\ldots, T]$. For state $x_t$, $\tilde{\mu}_t$ is found by solving $m$ minimizations for $\ell \in [1,\ldots, m]$ as follows:
\begin{equation}
    \tilde{\mu}^\ell_t(x_t)
      \in \argmin_{u^\ell_t\in \mathbf{U}^\ell_t(x_t)} E [g_t(x_t, u_t, \eta, \rho, \delta) + \tilde{J}_{{\pi},t+1}(x_{t+1})],
    \label{eq:Bellman_eq_ma}
\end{equation}
where $u_t=(\tilde{\mu}^1_t(x_t):\tilde{\mu}^{\ell-1}_t(x_t), u_t^\ell, {\mu}^{\ell+1}_{t}(x_t):{\mu}^{m}_{t}(x_t))$, and $\tilde{J}_{{\pi},t+1}(x_{t+1}) = |\mathbf{\overline{r}_{t+1+t_h}}| + \sum_{t' = t+1}^{t+t_h} g_{t'}(x_{t'}, {\mu}_{t'}(x_{t'}), \eta, \rho, \delta)$ is a cost approximation derived from $t_h$ applications of the base policy ${\pi}$ from state $x_{t+1}$, with a terminal cost approximation %$\hat{J}_{t+1+t_h} = 
$|\mathbf{\overline{r}_{t+1+t_h}}|$. 

To apply one-at-a-time rollout to a large city-scale problem, we design an algorithm that approximates this rollout scheme, but incurs a lower computational cost that satisfies user defined computational constraints. Our algorithm is given in Sec.~\ref{sec:our_approach}. We find a sufficiently large fleet size $m$ for which a reasonable base policy $\pi$ is stable, such that $E[Z_{\pi, T}] \leq m \cdot T$, as defined in Sec.~\ref{sec:stability_definition}. In particular, we are interested in the stability of the policy $\pi_{\text{base}}$ associated with IA-RA, as this policy is 2-competitive \cite{gerkey2004formal} and hence our approximate rollout approach obtains a near-optimal policy. 

\section{Approximation algorithm for multiagent rollout}
\label{sec:our_approach}
In this section we propose an approximate algorithm for multiagent rollout (see Eq.~\ref{eq:Bellman_eq_ma}). Our proposed method is composed of a two-phase planning scheme that reduces the computational cost of one-at-a-time rollout through partitioning of the map using the demand distribution. We take into account user defined computational constraints in the form of the maximum number of taxis that can be run by one-at-a-time rollout $m_{\text{lim}}$, and the length of the planning horizon $t_h$ (longer planning horizon result in longer runtimes). The algorithm is detailed in Algorithm \ref{alg:2phase_planner}. The proposed two-phase algorithm also takes as input $m$ the total number of taxis in the fleet. We provide theoretical bounds on $m$ in Sec.~\ref{sec:analysis} and calculated values in practice in Sec.~\ref{subsec:choice_m}.

The first routine in Algorithm \ref{alg:2phase_planner} is denoted as $get\_partitions$ and it places the center of each partition on the map. $get\_partitions$ solves a capacitated facility location problem \cite{WU2006}, where the capacity for each partition center is set to be $m_{\text{lim}}$, and then the expected number of requests for the ride service during the entire time horizon is used as the demand.
The $get\_partitions$ routine then assigns each node to the closest partition center using weighted $k$-means, where the weights of the nodes are given by the probability distribution of pickups. This routine guarantees that the size of each partition is inversely proportional to the density of requests.

After obtaining the partitions, Algorithm \ref{alg:2phase_planner} executes two routines at each time step: $High\_level\_planner$ (see Alg.~\ref{alg:high_level_planner}), and $Low\_level\_planner$ (see Alg.~\ref{alg:low_level_planner}). Intuitively, the $High\_level\_planner$ re-balances the taxis between partitions using an instantaneous assignment of taxis to current and expected future requests for the next $t_h$ time-steps as given by a certainty equivalence approximation. It returns the controls for taxis that are expected to go across regions $u_{h}^{g}$, as well as the list of high level taxis $\hat{m}$, and $\hat{d}$ the set of locations for the high level taxis to move towards. The $Low\_level\_planner$, on the other hand, plans for routing and pickup actions for taxis that remain in their original sectors according to the high level planner, executing one-at-a-time rollout with base policy IA-RA as defined in Eq.~\ref{eq:Bellman_eq_ma} to obtain $\tilde{u}_{t}^{k}$ the control of taxis in sector $k$ at time $t$.

After partitioning the graph, the state $x_t$ consists of $K$ sub-states $\{x^k_t\}_{k=1}^K$, one corresponding to a partition $k=\{1,\ldots,K\}$ of the graph. The state transition of partition $k$ is given by, 
$x^k_{t+1} = f^k(x^k_t, u^k_t, u^g_h(t, k), \eta, \rho, \delta).$
The control $u_t$ can be separated as $\{u^k_t, u^{g}_h(t, k)\}_{k=1}^K$, where the control component $u^k_t$ corresponds to the taxis that are local to partition $k$. %, that includes moving inside partition $k$ or pickup a request at partition $k$. 
The control component $u^{g}_h(t, k)$ corresponds to the controls of taxis coming into partition $k$ as given by the higher level planner. Since we consider the length of outstanding requests as the stage cost, we have $|\mathbf{\overline{r}_{t}}| = g_t(x_t, u_t, \eta, \rho, \delta)=\sum_{k=1}^K g^k_t(x^k_t, u^k_t, u^{g}_h(t, k), \eta, \rho, \delta) = \sum_{k=1}^K |\mathbf{\overline{r}^k_{t}}|,$ where $\mathbf{\overline{r}^k_{t}} \subseteq\mathbf{\overline{r}_{t}}$, and $\forall r\in \mathbf{\overline{r}^k_{t}},\rho_r\in s_k$.
The cost of our two-phase policy $\pi_{2P}$ is 
    \begin{align*}
        J_{\pi_{2P}}(x_1) & = E[\sum_{t=1}^{T} \sum_{k=1}^K g^k_{t}(x^k_{t}, \tilde{u}^k_{t}, u^{g}_h(t, k), \eta, \rho, \delta)] 
    \end{align*}

\begin{algorithm}[!ht]
\label{alg:2phase_planner}
\DontPrintSemicolon
  \KwInput{Initial state $x_1$, maximum number of taxis per sector $m_{\text{lim}}$, fleet size $m$, planning horizon $t_h$}
  \KwOutput{policy $\pi_{2P}$ that gives routing/pickup strategy for all taxis in the system}
  \caption{Two-phase Planner}

$K\gets \frac{m}{m_{\text{lim}}}$

$\{s_k\}_{k=1}^K \gets get\_partitions(m_{\text{lim}}, K, G, \eta, \rho, \delta)$ 

%$T_H \gets T / t_h$
    $\hat{d}\gets\{\}, \hat{m}\gets[]$

  \For{each time planning step $t\in[1,\dots, T]$}{
  
  $ u^{g}_{h}, \hat{m}, \hat{d} \gets High\_level\_planner(x_t, \eta, \rho, \delta, $\hfill$\pi_{\text{base}}, t_h, \hat{m}, \hat{d}, \{s_k\}_{k=1}^K)$

  %\For{each time step $t\in [h\cdot t_h, \dots, (h+1)\cdot t_h]$}{
    \For{each sector $s_k, k\in\{1, \ldots, K\}$ \text{in parallel}}{
        $\tilde{u}^k_t\gets Low\_level\_planner(x^k_t, u^{g}_{h}(t, k), \eta, \rho, $\hfill$\delta, \pi_{\text{base}}, t_h, \hat{m}, \hat{d})$
    }
    
    set $\mu_{2P,t}(x_t) = \{\tilde{u}^k_t, u^{g}_{h}(t, k)\}_{k=1}^K$
    
    $x_{t+1} \sim f(x_t, \mu_{2P,t}(x_t), \eta, \rho,\delta)$
  
  }
  set $\pi_{2P} = \{\mu_{2P,t} \}_{t=1}^{T}$

  return $\pi_{2P}$

\end{algorithm}

\begin{algorithm}[!ht]
\label{alg:high_level_planner}
\DontPrintSemicolon
  \KwInput{$x_t, \eta, \rho, \delta, \pi_{\text{base}}, t_h, \hat{m}, \hat{d}, \{s_k\}_{k=1}^K$}
  \KwOutput{control $u^{g}_{h}$ for taxis in $\hat{m}$, $\hat{m}$ the list of high level taxis, $\hat{d}$ the list of high level destinations for the high level taxis}
  \caption{High\_level\_planner}

  $r_{CE}\gets$ list of $t_h \cdot E[\eta]$ future requests emulated by certainty equivalence for the next $t_h$ steps
  %Obtain certainty equivalence representation for $t_h \cdot E[\eta]$ requests using expected pickup $E[\rho]$ and dropoff $E[\delta]$. This process will emulate the requests that will enter the system in the next $t_h$ time steps. %$E[\eta]$ requests will enter the system at each time step, with requests located in nodes with higher probability entering first.

  %\red{At} each time step $t\in[h\cdot t_h, \dots, (h+1)\cdot t_h]$ perform instantaneous assignment to determine which taxis are better suited to service which requests.

   $A\gets$ instantaneous assignment for the taxi set $\{\ell|\ell\in\{1,\ldots,m\}, \tau_\ell=0\}$ to requests $\{\mathbf{\overline{r}_{t}}\cup r_{CE}\}$
  %Perform an instantaneous assignment with the emulated future.

  %\red{Taxis} that after the last time step in $[h\cdot t_h, \dots, (h+1)\cdot t_h]$ have picked up a request in a different sector that the one where they started, or have moved outside their original sector to potential service another request in a future time step are promoted to high level taxis and are appended to list $\hat{m}$.

    Define $sector(v)\in \{k| k\in\{1,\ldots, K\}, v\in s_k\}, \forall v\in V$

    Define $next\_hop\_in\_partition(v_{start}, v_{end})\in \argmin_{v\in  s_{sector(v_{end})}\land v\in shortest\_path(v_{start}, v_{end})}$ $d(v_{start}, v)$, $\forall v_{start},v_{end}\in V, sector(v_{start})\neq sector(v_{end})$

    \For{each $(\ell,r) \in A,$}{
    
    \If{$sector(\nu_\ell)\neq sector(\rho_r)$}{
    
     $\hat{m}\gets \hat{m}\cup \{\ell\}$,
     
     $\hat{d}[\ell]\gets next\_hop\_in\_partition(\nu_\ell, \rho_r)$
     }
     
    }

    %taxis that need to move out of their current sector according to the instantaneous assignment are promoted to high level taxis and add it to list $\hat{m}$

  \For{each taxi $\ell \in \hat{m}$}{

  $u_h^g(t', k)\gets u_h^g(t', k)\cup v,$ $|v\in shortest\_path(v_\ell, \hat{d}[\ell])[t'], v\in s_k,  0<t'\leq d(v_\ell, \hat{d}[\ell])$

  %$u_h^g(t', k)\gets u_h^g(t', k)\cup \hat{d}[\ell],$ where $t'=d(\nu_\ell, \hat{d}[\ell])$
  
  %$d\gets$ pickup location of taxi $\ell$.
  
  %find shortest path from the taxi and the pickup location of the request. The global control for this taxi is each of the steps in the given shortest path.
    %If taxi picked up request at a sector different from its original sector, find shortest path from original location of the taxi and the request that the taxi picked up. The global control for this taxi is each of the steps in the given shortest path.

    %If taxi moved to a different sector but did not pick up a request, then find shortest path from the original location of the taxi and \red{the current location} of the taxi. The global control for this taxi is each of the steps in the given shortest path.
  }
  
  return $u^{g}_{h}$, $\hat{m}$, $\hat{d}$

\end{algorithm}

\begin{algorithm}[!ht]
\label{alg:low_level_planner}
\DontPrintSemicolon
  \KwInput{$x^k_t, u^{g}_{h}(t, k), \eta, \rho,\delta, \pi_{\text{base}}, t_h, \hat{m}, \hat{d}$}
  \KwOutput{$\tilde{u}^k_t$}
  \caption{Low\_level\_planner}

    %Find taxis inside the sector $s_k$ as $m^k\subseteq \{1,\ldots, m\},\forall \ell\in m^k, \nu_\ell\in s_k, \ell \notin \hat{m}$

    Find taxis inside the sector $s_k$ as $m^k=\{\ell| \ell\in\{1,\ldots, m\}, v_\ell\in s_k, \ell\notin \hat{m}\}$

    %Now we re-index the taxi numbers in $m^k$ using a map $I$, such that $I=[m^k[1],\ldots, m^k[|m^k|]]$
    
    Execute one-at-a-time rollout to obtain the control for each of the remaining taxis in $m^k$% in the current sector that are not currently being assigned a global control.
    
    %The local control is found following:
    Starting from, $\ell=[ 1,\ldots, |m^k|]$
  \begin{equation}
  \begin{aligned}
      %\tilde{u}^k_t\in \argmin_{u^k_t} E [g^k_t(x^k_t, u^k_t, u^{g}_{h}(t), \eta, \rho, \delta) + \tilde{J}(x^k_{t+1})]
      &\tilde{u}^{k,m^k[\ell]}_t
      \in \argmin_{u\in U_t^{m^k[\ell]}(x_t^k)} E [g^k_t(x^k_t, u', 
       u^{g}_{h}(t, k), \eta, \rho, \delta) \\&+ \tilde{J}^k(x^k_{t+1})],
      \label{eq:low_level_rollout}
  \end{aligned}
  \end{equation}
  where 
  $u'=(\tilde{u}^{k,m^k[1]}_t:\tilde{u}^{k,m^k[\ell-1]}_t$,$ u, u^{k,m^k[\ell+1]}_{base,t}:u^{k,m^k[|m^k|]}_{base,t})$ and  $u^{k,\ell}_{base,t}$ is the control given by the base policy $\pi_{base}^k$ in the local state $x_t^k$, and the next state is $x^k_{t+1} = f^k(x^k_t, u', 
       u^{g}_{h}(t, k), \eta, \rho, \delta)$
. The variable $\tilde{J}^k(x^k_{t+1}) $ is the %shorthanded version of 
    %$\tilde{J}^k_{t+1:t_h}(x^k_{t+1}) = E[\sum_{t'=t+1}^{t+t_h} g^k_{t'}(x^k_{t'}, u^k_{base, t'}, u^{g}_h(t, k), \eta, \rho, \delta)+ \hat{J}(x_{t+t_h})]$, where $\hat{J}(x_{t+t_h})$ 
    cost of applying the base policy $t_h$ times followed by a terminal cost approximation.

  return $\tilde{u}^k_t = [\tilde{u}^{k,m^k[1]}_t, \ldots, \tilde{u}^{k,m^k[\ell]}_t]$

\end{algorithm}
\begin{figure}
    \centering
    \vspace{5pt}
    \includegraphics[width=\linewidth]{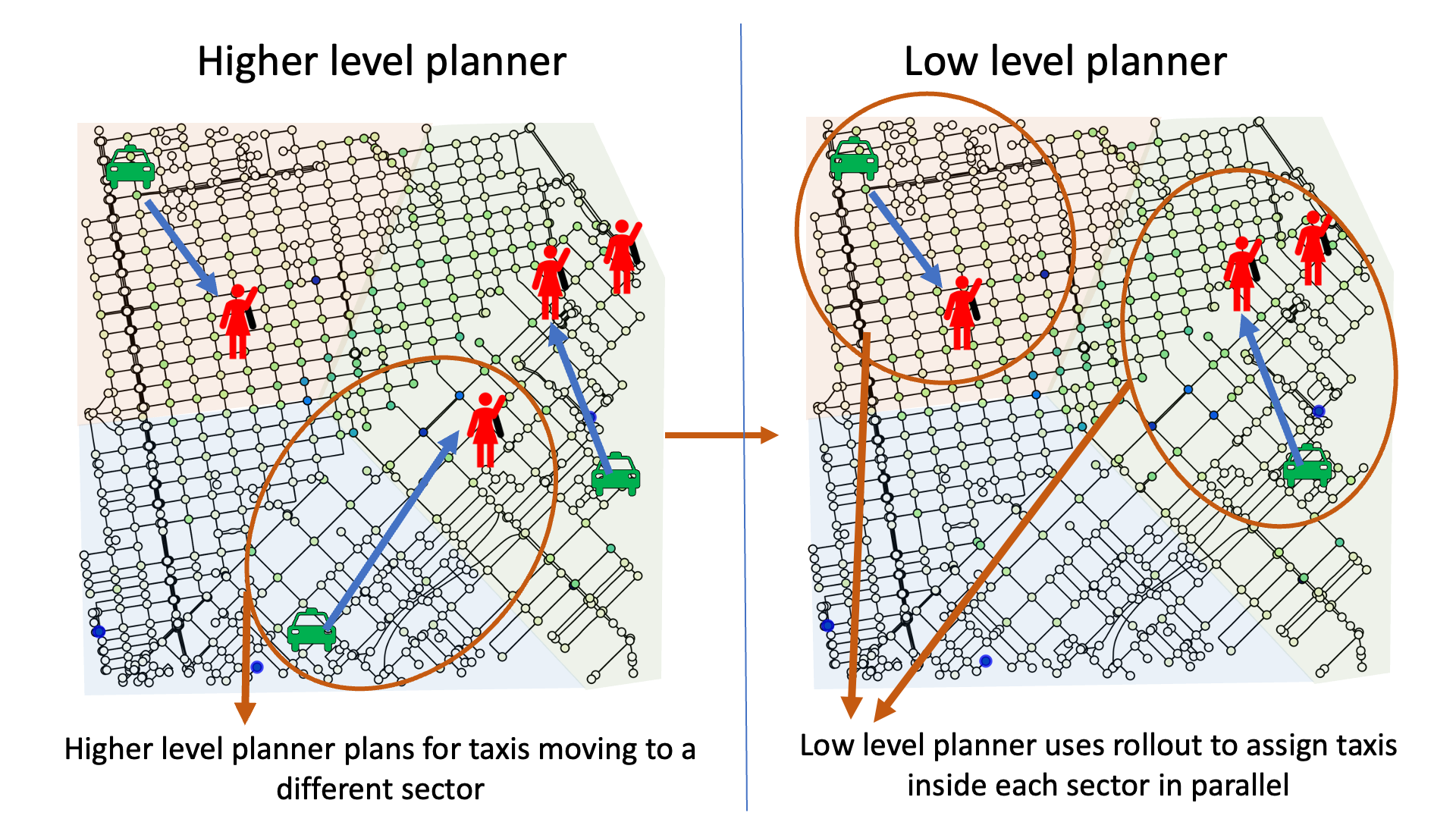}
    \caption{Our two phased approach executed on a map with 3 sectors.}
    \label{fig:two_phased_approach}
\end{figure}

Figure~\ref{fig:two_phased_approach} shows the two phased approach with an example with $3$ taxis and $4$ outstanding requests.

\section{Theoretical Results}
\label{sec:analysis}
In this section, we provide a sufficient condition for choosing a fleet size $m$ that will make the policy $\pi_{\text{base}}$, instantaneous assignment with reassignment (IA-RA), a stable policy. %such that $E[Z_{\pi_{\text{base}, T}}] < m \cdot T$. 
We also provide an asymptotic necessary condition on $m$ for the stability of $\pi_{\text{base}}$ as $T \to \infty$.

\subsection{Sufficient condition for stability of $\pi_{\text{base}}$}
\label{subsec:fleet_size_stability}
We are interested in finding the sufficient conditions on the fleet size $m$ that guarantee the stability of policy $\pi_{\text{base}}$ such that the relation $E[Z_{\pi_{\text{base}}, T}] \leq m \cdot T$ always holds. To do so, we first analyze the policy $\hat{\pi}$ referred to as random instantaneous assignment, where taxis are randomly assigned to requests. Under this policy, a taxi does not move until it has been assigned to a request. Once a taxi is assigned to a request, the taxi cannot be assigned to other requests until it has serviced the originally assigned request. By having a random assignment of requests to taxis, $l_{r_q, \hat{\pi}}$ for an arbitrary request $r_q$ becomes a random variable instead of a deterministic function of the requests in the system and the locations of all the taxis. The randomness in $\hat{\pi}$ also makes the request's pickup location $\rho_{r_q}$ and the location of the taxi assigned to the request $l_{r_q, \hat{\pi}}$ independent, making the analysis easier. Using this policy, we can find an upper bound on $E[Z_{\hat{\pi}, T}]$, and choose $m$ such that $m \cdot T$ is greater than or equal to the upper bound, making $\hat{\pi}$ a stable policy by definition. 
We then show that the IA-RA policy $\pi_{\text{base}}$ given by a matching algorithm, like the auction algorithm \cite{Bertsekas2020Auction} or the modified JVC algorithm \cite{Crouse2016}, results in a smaller expected service distance than $\hat{\pi}$, i.e., $E[Z_{\pi_{\text{base}}, T}] \leq E[Z_{\hat{\pi}, T}]$. This implies $E[Z_{\pi_{\text{base}}, T}] \leq E[Z_{\hat{\pi}, T}] \leq m \cdot T$, and hence $\pi_{\text{base}}$ constitutes a stable policy for the sufficiently large fleet size $m$ found in the analysis of the stability of $\hat{\pi}$. We present the formal claim for the sufficient conditions on $m$ for the stability of $\hat{\pi}$ below in the following lemma.

\begin{lemma}
    Let the random variable $l_{\text{rand}}$ with support $V$ represent the location of a random taxi that gets assigned to a request after that taxi has previously served a different request. Define $D_{\text{max}} \triangleq \max\{E[d(\xi, \rho)], E[d(l_{\text{rand}}, \rho)]\} + E[d(\rho, \delta)]$. If the fleet size $m$ satisfies \begin{equation}
        m \geq  E_{\eta}[\eta] \cdot D_{\text{max}}, \nonumber
    \end{equation}
    then the policy associated with a random instantaneous assignment of taxis to requests, $\hat{\pi}$, constitutes a stable policy such that $E[Z_{\hat{\pi}, T}] \leq m \cdot T$.
    \label{lemma:stability_random}
\end{lemma}
    
\begin{proof}
   First we rewrite $Z_{\hat{\pi}, T}$ to consider the index of the requests irrespective of the time at which those requests enter the system. This allows us to get:
    \begin{align}
        Z_{\hat{\pi}, T} = \sum_{q=1}^{R_T} W_{r_q,\hat{\pi}} \nonumber
    \end{align}
    where $R_T = \sum_{t=1}^{T} \eta_{t}$. Under the policy $\hat{\pi}$, a request is randomly assigned to a taxi. For this reason, there is a possibility that not all taxis in the fleet end up servicing a request. Hence, we define $\bar{m}$ as the random variable that corresponds to the effective fleet size over the entire time horizon, i.e., the number of taxis that pickup at least one request. We reindex requests such that requests that are assigned to taxis that haven't serviced any requests come before requests that are assigned to taxis that have already serviced one or more requests. This means that the first $\bar{m}$ requests after reindexing are serviced by taxis that are still at their original initial locations. The rest of the requests are then serviced by taxis that are at the dropoff location of their previously serviced request. We can therefore, rewrite $Z_{\hat{\pi}, T}$ as follows:
    \begin{align}
        Z_{\hat{\pi}, T} = \sum_{a=1}^{\bar{m}} W_{r_a, \hat{\pi}} + \sum_{b=\bar{m}+1}^{R_T} W_{r_b, \hat{\pi}} \nonumber
    \end{align}
    Where, we have $W_{r_a, \hat{\pi}} = d(l_{r_a, \hat{\pi}}, \rho_{r_a}) + d(\rho_{r_a}, \delta_{r_a})$. Similarly, $W_{r_b, \hat{\pi}} = d(l_{r_b, \hat{\pi}}, \rho_{r_b}) + d(\rho_{r_b}, \delta_{r_b})$. The random variable $l_{r_a, \hat{\pi}}$ has the same distribution as $\xi$ since under $\hat{\pi}$ taxis are randomly matched to requests and taxis that haven't been assigned to a request are still at their initial locations. Similarly, $l_{r_b, \hat{\pi}}$ has the same distribution as $l_{\text{rand}}$ since taxis are randomly matched to requests and taxis that have already serviced at least one request are located at the drop-off locations of their previously serviced request. 
    
    We define $W_{r_{\text{init}}, \hat{\pi}} = d(l_{r_{\text{init}, \hat{\pi}}}, \rho_{r_{\text{init}}}) + d(\rho_{r_{\text{init}}}, \delta_{r_{\text{init}}})$, where $l_{r_{\text{init}, \hat{\pi}}}$ is a random variable with the same distribution as $\xi$, $\rho_{r_{\text{init}}}$ is a random variable with the same distribution as $\rho$, and $\delta_{r_{\text{init}}}$ is a random variable with the same distribution as $\delta$. Now, we define $W_{r_{\text{next}}, \hat{\pi}} = d(l_{r_{\text{next}, \hat{\pi}}}, \rho_{r_{\text{next}}}) + d(\rho_{r_{\text{next}}}, \delta_{r_{\text{next}}})$, where $l_{r_{\text{next}, \hat{\pi}}}$ is a random variable with the same distribution as $l_{\text{rand}}$, $\rho_{r_{\text{next}}}$ is a random variable with the same distribution as $\rho$, and $\delta_{r_{\text{next}}}$ is a random variable with the same distribution as $\delta$. We can see that all $W_{r_a, \hat{\pi}}$ have the same distribution as $W_{r_{\text{init}}, \hat{\pi}}$ since $W_{r_a, \hat{\pi}}$ is a function of three independent random variables $l_{r_a, \hat{\pi}}, \rho_{r_a},$ and $\delta_{r_a}$ with the same distributions as $\xi, \rho$ and $\delta$, respectively. Similarly, all $W_{r_b, \hat{\pi}}$ have the same distribution as $W_{r_{\text{next}}, \hat{\pi}}$ since $W_{r_b, \hat{\pi}}$ is a function of three independent random variables $l_{r_b, \hat{\pi}}, \rho_{r_b},$ and $\delta_{r_b}$ with the same distributions as $l_{\text{rand}}, \rho$ and $\delta$, respectively. Using these two facts, we will obtain an upper bound for $E[Z_{\hat{\pi}, T}]$. Then, we set $m \cdot T$ to be greater than or equal to the upper bound to satisfy the definition of stability given in Sec.~\ref{sec:stability_definition}. We consider the following:
    \begin{align}
        & E[Z_{\hat{\pi}, T} ] = E\left[ \sum_{a=1}^{\bar{m}} W_{r_a, \hat{\pi}} + \sum_{b=\bar{m}+1}^{R_T} W_{r_b, \hat{\pi}} \right] \nonumber \\
        & \overset{(1)}{=} E \left[ E\left[ \sum_{a=1}^{\bar{m}} W_{r_a, \hat{\pi}} \middle| \bar{m} \right] \right]  + E \left[ E\left[ \sum_{b=\bar{m}+1}^{R_T} W_{r_b, \hat{\pi}} \middle| R_T, \bar{m} \right]\right] \nonumber \\
        & \overset{(2)}{=} E \left[\sum_{a=1}^{\bar{m}} E[W_{r_a, \hat{\pi}} | \bar{m}] \right] + E \left[ \sum_{b=\bar{m}+1}^{R_T} E\left[W_{r_b, \hat{\pi}} \middle| R_T, \bar{m} \right]\right] \nonumber \\
        & \overset{(3)}{=} E \left[ \sum_{a=1}^{\bar{m}} E[W_{r_{\text{init}}, \hat{\pi}} ] \right] + E \left[ \sum_{b=\bar{m}+1}^{R_T} E\left[W_{r_{\text{next}}, \hat{\pi}} \right]\right] \nonumber \\
        & = E[\bar{m}] \cdot E[W_{r_{\text{init}}, \hat{\pi}}] + E \left[ (R_T-\bar{m}) \cdot E\left[W_{r_{\text{next}}, \hat{\pi}} \right]\right] \nonumber \\
        & = E[\bar{m}] \cdot E[W_{r_{\text{init}}, \hat{\pi}} ] + (E[R_T]-E[\bar{m}]) \cdot E[W_{r_{\text{next}}, \hat{\pi}} ] \nonumber
    \end{align}
    Where equality (1) comes from linearity of expectations and the law of total expectations; equality (2) comes from linearity of the conditional expectations; equality (3) comes from the fact that $W_{a, \hat{\pi}}$ is independent of $\bar{m}$, $W_{a, \hat{\pi}}$ has the same distribution as $W_{r_{\text{init}}, \hat{\pi}}$, $W_{r_{b}, \hat{\pi}}$ is independent of $R_T$ and $\bar{m}$, and $W_{r_{b}, \hat{\pi}}$ has the same distribution as $W_{r_{\text{next}}, \hat{\pi}}$ as explained before.
    
    We can obtain an upper bound for $E[W_{r_{\text{init}}, \hat{\pi}}]$ as follows:
    \begin{align}
        \label{eq:upper_bound_init}
        E[W_{r_{\text{init}}, \hat{\pi}}] & \overset{(1)}{=} E[d(l_{r_{\text{init}, \hat{\pi}}}, \rho_{r_{\text{init}}})] + E[d(\rho_{r_{\text{init}}}, \delta_{r_{\text{init}}})] \nonumber \\
        & \overset{(2)}{=} E[d(\xi, \rho)] + E[d(\rho, \delta)] \nonumber \\
        & \leq \max \{ E[d(\xi, \rho)], E[d(l_{\text{rand}}, \rho)]\} + E[d(\rho, \delta)] \nonumber \\
        & \overset{(3)}{=} D_{\text{max}}
    \end{align}
    Where equality (1) comes from linearity of expectations, equality (2) comes from the definition of $W_{r_{\text{init}}, \hat{\pi}}$, and equality (3) comes from the definition of $D_{\text{max}}$. Similarly, we can obtain an upper bound for $E[W_{r_{\text{next}}, \hat{\pi}}]$ as follows:
    \begin{align}
        \label{eq:upper_bound_next}
        E[W_{r_{\text{next}}, \hat{\pi}}] & \overset{(1)}{=} E[d(l_{r_{\text{next}, \hat{\pi}}}, \rho_{r_{\text{next}}})] + E[d(\rho_{r_{\text{next}}}, \delta_{r_{\text{next}}})] \nonumber \\
        & \overset{(2)}{=} E[d(l_{\text{rand}}, \rho)] + E[d(\rho, \delta)] \nonumber \\
        & \leq \max \{ E[d(\xi, \rho)], E[d(l_{\text{rand}}, \rho)]\} + E[d(\rho, \delta)] \nonumber \\
        & \overset{(3)}{=} D_{\text{max}}
    \end{align}
    Where equality (1) comes from linearity of expectations, equality (2) comes from the definition of $W_{r_{\text{next}}, \hat{\pi}}$, and equality (3) comes from the definition of $D_{\text{max}}$. 
    Using these two results, we can upper bound $E[Z_{\hat{\pi}, T}]$ as follows:
    \begin{align}
        E[Z_{\hat{\pi}, T}] & \overset{(1)}{\leq} E[\bar{m}] \cdot D_{\text{max}} + (E[R_T]-E[\bar{m}]) \cdot D_{\text{max}} \nonumber \\
        & = E[R_T] \cdot D_{\text{max}} \nonumber \\
        & \overset{(2)}{=} E\left[ \sum_{t=1}^{T} \eta_t \right] \cdot D_{\text{max}} \nonumber \\
        & \overset{(3)}{=} \sum_{t=1}^{T} E[\eta] \cdot D_{\text{max}} \nonumber \\
        & = T \cdot E[\eta] \cdot D_{\text{max}} \nonumber
    \end{align}
    where inequality (1) comes from and application of the upper bounds in Equation~\ref{eq:upper_bound_init} and Equation~\ref{eq:upper_bound_next}, equality (2) comes from the definition of $R_T$, and equality (3) comes from the linearity of expectations and the fact that variables $\eta_t$ are independent and identically distributed. From the stability definition given in Sec.~\ref{sec:stability_definition}, $\hat{\pi}$ is stable as long as $E[Z_{\hat{\pi}, T}] \leq m \cdot T$. Therefore, if we choose $m \cdot T \geq T \cdot E[\eta] \cdot D_{\text{max}}$ and hence $m \geq E[\eta] \cdot D_{\text{max}}$, this would be a sufficient condition to guarantee that the random instantaneous assignment policy $\hat{\pi}$ is stable such that the relation $E[Z_{\hat{\pi}, T}] \leq m \cdot T$ holds.
\end{proof}

Notice that we can express the probability distribution $p_{l_{\text{rand}}}$ for $l_{\text{rand}}$ using the marginalization of the probabilities of the previous pickup-dropoff combinations, and hence all the terms given in $D_{\text{max}}$ can be calculated in practice using historical data. We use the result from \cref{lemma:stability_random} to show that the same $m$ chosen to guarantee stability of $\hat{\pi}$ serves as a sufficiently large $m$ to guarantee stability of $\pi_{\text{base}}$ which is formalized in the Theorem~\ref{theorem:standard_IA_stability}.

% the upper bound using data, we know that the random variable $l_{\text{rand}}$ corresponds to the dropoff of the previous serviced request, which depends on the pickup of such request. We denote the previously serviced request as $r_{\text{prev}}$ with pickup $\rho_{r_{\text{prev}}}$ and dropoff $\delta_{r_{\text{prev}}}$. From this, we can express the probability distribution $p_{l_{\text{rand}}}$ for $l_{\text{rand}}$ using the marginalization of the probabilities of the previous pickup-dropoff combinations. Hence, we get that $Pr(l_{\text{rand}} = i)$, the probability of the assigned taxi being at an arbitrary node $i \in V$, is: 
%     \begin{align*}
%         Pr(l_{\text{rand}} = i) & =  \sum_{j \in V} Pr (l_{\text{rand}} =i, \rho_{r_{\text{prev}}} = j) \\
%         & =  \sum_{j \in V} Pr ( \delta_{r_{\text{prev}}} = i, \rho_{r_{\text{prev}}} = j) \\
%         & =  \sum_{j \in V} Pr ( \delta_{r_{\text{prev}}} = i | \rho_{r_{\text{prev}}} = j) \cdot Pr(\rho_{r_{\text{prev}}} = j) \\
%     \end{align*}
%     We can now estimate this probability distribution using data.

\begin{theorem}
    Assume that the fleet size $m$ satisfies the condition given in \cref{lemma:stability_random}. Then the policy $\pi_{\text{base}}$, which corresponds to standard instantaneous assignment with reassignment at each time step (IA-RA), is a stable policy such that $E[Z_{\pi_{\text{base}}, T}] \leq m \cdot T$, for a finite horizon $T>0$.
    \label{theorem:standard_IA_stability}
\end{theorem}

\begin{proof}
    To prove this statement, we will show that random instantaneous assignment $\hat{\pi}$ results in longer expected distance traveled per assigned request than $\bar{\pi}$ standard instantaneous assignment with commitment to the initial assignment such that $ E[Z_{\bar{\pi}, T}] \leq E[Z_{\hat{\pi}, T}]$. Then we will show that $\bar{\pi}$ results in longer distance traveled per assigned request than $\pi_{\text{base}}$ standard instantaneous assignment with reassignment (IA-RA) at each time step such that $Z_{\pi_{\text{base}}, T} \leq Z_{\bar{\pi}, T}$, which implies $E[Z_{\pi_{\text{base}}, T}] \leq E[Z_{\bar{\pi}, T}]$. Since $\hat{\pi}$ is a stable policy for a fleet size of size $m \geq E[\eta] \cdot D_{\text{max}}$ according to lemma~\ref{lemma:stability_random}, and $Z_{\pi_{\text{base}}, T} \leq Z_{\bar{\pi}, T} \leq Z_{\hat{\pi}, T}$, then we can conclude that $\pi_{\text{base}}$ is also stable since $E[Z_{\pi_{\text{base}, T}}] \leq E[Z_{\bar{\pi}, T}] \leq E[Z_{\hat{\pi}, T}] \leq m \cdot T$. 
    
    We start by showing $ E[Z_{\bar{\pi}, T}] \leq E[Z_{\hat{\pi}, T}]$. To do this, we consider the following:
    \begin{align}
        & E[Z_{\hat{\pi}, T}] \overset{(1)}{=} E\left[ \sum_{t=1}^T \sum_{q=R_{t-1}+1}^{R_{t}} d(l_{r_q,\hat{\pi}}, \rho_{r_q}) + d(\rho_{r_q}, \delta_{r_q}) \right] \nonumber \\
        & \overset{(2)}{=} \sum_{t=1}^T E\left[ \sum_{q=R_{t-1}+1}^{R_{t}} d(l_{r_q,\hat{\pi}}, \rho_{r_q}) + d(\rho_{r_q}, \delta_{r_q}) \right] \nonumber \\
        & \overset{(3)}{\geq} \sum_{t=1}^T \min_{l_{R_{t-1}+1}, \dots, l_{R_{t}} } E \left[ \sum_{q=R_{t-1}+1}^{R_{t}} d(l_{r_q}, \rho_{r_q}) + d(\rho_{r_q}, \delta_{r_q}) \right] \nonumber \\
        & \overset{(4)}{\geq} \sum_{t=1}^T E \left[ \min_{l_{R_{t-1}+1}, \dots, l_{R_{t}} } \sum_{q=R_{t-1}+1}^{R_{t}} d(l_{r_q}, \rho_{r_q}) + d(\rho_{r_q}, \delta_{r_q}) \right] \nonumber \\
        & \overset{(5)}{=} \sum_{t=1}^T E \left[\sum_{q=R_{t-1}+1}^{R_{t}} W_{r_q,\bar{\pi}} \right] \nonumber \\
        & \overset{(6)}{=} E[Z_{\bar{\pi}, T}] \nonumber
    \end{align}

    Where equality (1) comes from the definition of $Z_{\hat{\pi}, T}$; equality (2) comes from linearity of expectations; inequality (3) comes from the definition of minimum expected traveled distance; inequality (4) comes from the argument presented in the Appendix (see Sec.~\ref{sec:appendix}); equality (5) follows from the definition of $W_{r_q,\bar{\pi}}$; and equality (6) follows from the definition of $Z_{\bar{\pi}, T}$. From this, we get $E[ Z_{\bar{\pi}, T}] \leq E[ Z_{\hat{\pi}, T}] \leq m \cdot T$, and hence we can conclude that $\bar{\pi}$ is a stable policy for a fleet of size given by lemma~\ref{lemma:stability_random}. 
    
    To prove that $Z_{\bar{\pi}, T} \geq Z_{\pi_{\text{base}}, T}$, we will consider an argument for two arbitrary requests $r_k$ and $r_h$. We provide an argument with two requests for simplicity, but it is important to note that this argument can be easily generalized to as many requests as needed. The distance associated with the assignments produced by $\bar{\pi}$ is $W_{r_k, \bar{\pi}} = d(l_{r_k,\bar{\pi}}, \rho_{r_k}) + d(\rho_{r_k}, \delta_{r_k})$ and $W_{r_h, \bar{\pi}} = d(l_{r_h,\bar{\pi}}, \rho_{r_h}) + d(\rho_{r_h}, \delta_{r_h})$. We assume that request $r_k$ enter the system before request $r_h$, but request $r_h$ enter the system at time $t'$ before the taxi assigned to request $r_k$ was able to pick request $r_k$ up. If we consider the time step $t'$, policy $\bar{\pi}$ will match request $r_h$ to one of the taxis that hasn't been assigned to any other request. Policy $\pi_{\text{base}}$, on the other hand, will perform a matching based on the distance from any available taxi to request $r_h$. In this sense, the original assignment given by $\bar{\pi}$ will be preserved unless a different assignment of free taxis to outstanding requests results in a lower distance. If we focus on the impact of the reassignment for these two requests at time $t'$, we get:
    \begin{align}
        & Z_{\bar{\pi}} = W_{r_k, \bar{\pi}} + W_{r_h, \bar{\pi}} \nonumber \\
        & \overset{(1)}{=} d(l_{r_k,\hat{\pi}}, \rho_{r_k}) + d(\rho_{r_k}, \delta_{r_k}) +  d(l_{r_h,\bar{\pi}}, \rho_{r_h}) + d(\rho_{r_h}, \delta_{r_h}) \nonumber \\
        & \geq \min_{l_{r_k}, l_{r_h}} d(l_{r_k}, \rho_{r_k}) + d(\rho_{r_k}, \delta_{r_k}) +  d(l_{r_h}, \rho_{r_h}) + d(\rho_{r_h}, \delta_{r_h}) \nonumber \\
        & \overset{(2)}{=} W_{r_k, \pi_{\text{base}}} + W_{r_h, \pi_{\text{base}}} \nonumber \\
        & \overset{(3)}{=} Z_{\pi_{\text{base}}} \nonumber
    \end{align}
    where equality (1) comes from the definition of $W_{r_k, \bar{\pi}}$ and $W_{r_h, \bar{\pi}}$, equality (2) comes from the definition of $\pi_{\text{base}}$, and equality (3) comes from the definition of $Z_{\pi_{\text{base}}, T}$.
    
    From this, we can conclude $Z_{\bar{\pi}, T} \geq Z_{\pi_{\text{base}}, T}$, and hence $ E[Z_{\pi_{\text{base}}, T}] \leq E[ Z_{\bar{\pi}, T}]  \leq m \cdot T$. Therefore, we can conclude that $\pi_{\text{base}}$ is a stable policy for a fleet size of $m$ as given by lemma~\ref{lemma:stability_random}.
\end{proof}

\subsection{Necessary condition for stability of $\pi_{\text{base}}$}

We are interested in finding the necessary condition for stability of policy $\pi_{\text{base}}$ asymptotically as $T \to \infty$. For this reason, we want to find a lower bound on $E[Z_{\pi_{\text{base}}, T} / T]$. Choosing a fleet size $m$ smaller than this lower bound would make the policy $\pi_{\text{base}}$ asymptotically unstable, i.e., $E[Z_{\pi_{\text{base}}, T}] >  m \cdot T$ as $T \to \infty$. To obtain this result, we first find a lower bound for $E[Z_{\pi_{\text{base}}, T} / T]$, the expected travel distance associated with servicing the requests that enter the system per time step, and then we apply a limit as $T \to \infty$ to obtain an expression for the asymptotic lower bound. The following theorem states this result formally.

\begin{theorem}
    Let $\mathit{WD}(p_\delta, p_\rho)$ denote the first Wasserstein distance \cite{ruschendorf1985wasserstein} between probability distributions $p_{\delta}$ and $p_{\rho}$ with support $\Omega$, such that:
    \begin{align*}
        \mathit{WD}(p_\delta, p_\rho) = \inf_{\gamma \in \Gamma(p_\delta, p_\rho)} \int_{x,y\in \Omega} ||y-x|| d\gamma(x,y),
    \end{align*}
    where $||\cdot||$ is the euclidean metric, and $\Gamma(p_\delta, p_\rho)$ is the set of measures over the product space $\Omega \times \Omega$ having marginal densities $p_\delta$ and $p_{\rho}$, respectively. Define $D_{\text{min}} \triangleq \mathit{WD}(p_\delta, p_\rho) + E[d(\rho, \delta)]$. Assume that the random variables for pickups $\rho$ and drop-offs $\delta$ are independent and we have a fleet of size $m < E[\eta] \cdot D_{\text{min}}$. Then, the policy $\pi_{\text{base}}$ is asymptotically unstable, i.e., $E[Z_{\pi_{\text{base}}, T}] >  m \cdot T$ as $T \to \infty$.
    \label{theorem:min_num_taxis}
\end{theorem}

\begin{proof}
    We denote $\pi_{\text{base}}$ as the policy that results from  standard instantaneous assignment with rematching at each time step. We can rewrite $Z_{\pi_{\text{base}}, T}$ as follows:
    \begin{align}
        Z_{\pi_{\text{base}}, T} & = \sum_{q=1}^{R_T} W_{r_q,\pi_{\text{base}}} \nonumber
    \end{align}
    where $R_T = \sum_{t=1}^{T} \eta_{t}$. We reindex requests such that requests that are assigned to taxis that haven't serviced any requests yet come before requests that are assigned to taxis that have already serviced one or more requests. We assume that our system has an effective fleet size of $\bar{m}$ (total number of taxis that get used). This means that the first $\bar{m}$ requests after reindexing are serviced by taxis that were originally at their initial locations before any assignment. The rest of the requests are then serviced by taxis that were originally at the dropoff location of their previously serviced request before any assignment. We can therefore, rewrite $Z$ as follows:
    \begin{align}
        Z_{\pi_{\text{base}}, T} & = \sum_{a=1}^{\bar{m}} W_{r_a, \pi_{\text{base}}} + \sum_{r_b=\bar{m}+1}^{R_T} W_{b, \pi_{\text{base}}} \nonumber
    \end{align}
    Where, we have $W_{r_a, \hat{\pi}} = d(l_{r_a, \pi_{\text{base}}}, \rho_{r_a}) + d(\rho_{r_a}, \delta_{r_a})$. Similarly, $W_{r_b, \hat{\pi}} = d(l_{r_b, \pi_{\text{base}}}, \rho_{r_b}) + d(\rho_{r_b}, \delta_{r_b})$. Random variable $l_{r_a, \pi_{\text{base}}}$ depends on the distribution of initial locations of the taxis $p_{\xi}$, while random variable $l_{r_b, \pi_{\text{base}}}$ depends on the distribution of drop-offs for the requests $p_{\delta}$.
    
    We are interested in finding an asymptotic lower bound for $E[Z_{\pi_{\text{base}}, T} /T ]$. We first consider:
    \begin{align}
        E \left[\frac{Z_{\pi_{\text{base}}, T}}{T} \right] & \overset{(1)}{=} E\left[ \frac{1}{T} \left( \sum_{a=1}^{\bar{m}} W_{r_a, \pi_{\text{base}}} + \sum_{b=\bar{m}+1}^{R_T} W_{r_b, \pi_{\text{base}}} \right) \right] \nonumber \\
        & \overset{(2)}{=} E\left[ \frac{1}{T} \sum_{a=1}^{\bar{m}} W_{r_a, \pi_{\text{base}}} \right] \nonumber \\ 
        & \qquad + E \left[ E\left[ \frac{1}{T} \sum_{b=\bar{m}+1}^{R_T} W_{r_b, \pi_{\text{base}}} \middle| R_T \right]\right] \nonumber \\
        & \overset{(3)}{=} \frac{1}{T} \sum_{a=1}^{\bar{m}} E[W_{r_a, \pi_{\text{base}}}] \nonumber \\ 
        & \qquad + E \left[ \frac{1}{T} \sum_{b=\bar{m}+1}^{R_T} E\left[W_{r_b, \pi_{\text{base}}} \middle| R_T \right]\right] \nonumber 
    \end{align}
    Where equality (1) comes from the definition of $Z_{\pi_{\text{base}}, T}$, equality (2) comes from linearity of expectations and the law of total expectations, and equality (3) comes from linearity of expectations and from the fact that given $R_T$, the conditional expectation can be moved inside the sum of the rightmost term.
    
    We define $\bar{d}_{\text{BMP}(\xi, \rho, \bar{m})}$ and $\bar{d}_{\text{BMP}(\delta, \rho, R_T)}$ as the average lengths for the optimal solutions of the bipartite matching problem (BMP) for a fleet size of $\bar{m}$ and $R_T$ taxis, respectively, origins distributed according to probability distributions $p_\xi$ and $p_\delta$, respectively, and destinations distributed according to $p_\rho$. Similarly, we define $\bar{d}_{\text{EBMP}(\xi, \rho, \bar{m})}$ and $\bar{d}_{\text{BMP}(\delta, \rho, R_T)}$ as the average lengths for the optimal solutions of the euclidean bipartite matching problem (EBMP) for a fleet size of $\bar{m}$ and $R_T$ taxis, respectively, origins distributed according to probability distributions $p_\xi$ and $p_\delta$, respectively, and destinations distributed according to $p_\rho$. Since $\bar{d}_{\text{BMP}(\xi, \rho, \bar{m})}$ is defined as the shortest path in the city graph between matched origins and destinations, we can easily see that this quantity can be lower bounded by $\bar{d}_{\text{EBMP}(\xi, \rho, \bar{m})}$ such that $\bar{d}_{\text{BMP}(\xi, \rho, \bar{m})} \geq \bar{d}_{\text{EBMP}(\xi, \rho, \bar{m})}$. Using this fact, we can lower bound $E[W_{r_a, \pi_{\text{base}}}]$ as follows:
    \begin{align}
        \label{eq:lower_bound_init}
        E[W_{r_a, \pi_{\text{base}}}] & \overset{(1)}{=} E[d(l_{r_a, \pi_{\text{base}}}, \rho_{r_a}) + d(\rho_{r_a}, \delta_{r_a})] \nonumber \\
        & \overset{(2)}{=} E[d(l_{r_a, \pi_{\text{base}}}, \rho_{r_a})] + E[d(\rho, \delta)] \nonumber \\
        & \overset{(3)}{\geq} \bar{d}_{\text{BMP}(\xi, \rho, \bar{m})} + E[d(\rho, \delta)] \nonumber \\
        & \overset{(4)}{\geq} \bar{d}_{\text{EBMP}(\xi, \rho, \bar{m})} + E[d(\rho, \delta)] \nonumber \\
        & = D_{\text{EBMP}}(\xi, \rho, \bar{m}, \delta)
    \end{align}
    With $D_{\text{EBMP}}(\xi, \rho, \bar{m},\delta) \triangleq \bar{d}_{\text{EBMP}(\xi, \rho, \bar{m})} + E[d(\rho, \delta)]$. Equality (1) comes from the definition of $W_{r_a, \pi_{\text{base}}}$; equality (2) comes from linearity of expectations and the fact that $\rho_{r_a}$ and $\delta_{r_a}$ have the same distribution as $\rho$ and $\delta$, respectively; inequality (3) comes from the fact that since BMP assumes that all requests enter the system at time 0, obtaining a solution for this problem results in a smaller distance traveled than when requests enter the system at different time steps and hence taxis need to be reassigned; inequality (4) comes from the fact that $\bar{d}_{\text{BMP}(\xi, \rho, \bar{m})} \geq \bar{d}_{\text{EBMP}(\xi, \rho, \bar{m})}$ as explained above.
    
    Similarly, we can lower bound $E\left[W_{r_b, \pi_{\text{base}}} \middle| R_T \right]$ for any $b > \bar{m}$ as follows:
    \begin{align}
        \label{eq:lower_bound_next}
        E[W_{r_b, \pi_{\text{base}}} | R_T] & \overset{(1)}{=} E[d(l_{r_b, \pi_{\text{base}}}, \rho_{r_b}) + d(\rho_{r_b}, \delta_{r_b}) | R_T] \nonumber \\
        & \overset{(2)}{=} E[d(l_{r_b, \pi_{\text{base}}}, \rho_{r_b}) | R_T] + E[d(\rho, \delta)] \nonumber \\
        & \overset{(3)}{\geq} \bar{d}_{\text{BMP}(\delta, \rho, R_T)} + E[d(\rho, \delta)] \nonumber \\
        & \overset{(4)}{\geq} \bar{d}_{\text{EBMP}(\delta, \rho, R_T)} + E[d(\rho, \delta)] \nonumber \\
        & = D_{\text{EBMP}}(\delta, \rho, R_T, \delta),
    \end{align}
    where $D_{\text{EBMP}}(\delta, \rho, R_T, \delta) = \bar{d}_{\text{EBMP}(\delta, \rho, R_T)} + E[d(\rho, \delta)]$. Equality (1) comes from the definition of $W_{r_b, \pi_{\text{base}}}$; equality (2) comes from linearity of expectations and the fact that $\rho_{r_b}$ and $\delta_{r_b}$ have the same distribution as $\rho$ and $\delta$, respectively; inequality (3) comes from the fact that since BMP assumes that all requests enter the system at time 0, obtaining a solution for this problem results in a smaller distance traveled than when requests enter the system at different time steps and hence taxis need to be reassigned; inequality (4) comes from the fact that the euclidean distance is not constrained to the structure of the graph and hence $\bar{d}_{\text{BMP}(\delta, \rho, R_T)} \geq \bar{d}_{\text{EBMP}(\delta, \rho, R_T)}$.
    
    Using these two lower bounds, we get:
    \begin{align}
    \label{eq:lower_bound_avg_work}
        E \left[\frac{Z_{\pi_{\text{base}}, T}}{T} \right] &  \overset{(1)}{\geq} \frac{1}{T} \sum_{a=1}^{\bar{m}} D_{\text{EBMP}}(\xi, \rho, \bar{m}, \delta) \nonumber \\ 
        & \qquad + E \left[ \frac{1}{T} \sum_{b=\bar{m}+1}^{R_T} D_{\text{EBMP}}(\delta, \rho, R_T, \delta) \right] \nonumber \\
        & =  \frac{\bar{m}}{T} D_{\text{EBMP}}(\xi, \rho, \bar{m}, \delta) \nonumber \\ 
        & \qquad + E \left[ \frac{1}{T} (R_T-\bar{m}) D_{\text{EBMP}}(\delta, \rho, R_T, \delta) \right] \nonumber\\
        & \overset{(2)}{=} \frac{\bar{m}}{T} D_{\text{EBMP}}(\xi, \rho, \bar{m}, \delta) \nonumber \\  
        & \qquad - \frac{\bar{m}}{T} E[D_{\text{EBMP}}(\delta, \rho, R_T, \delta)] \nonumber \\
        & \qquad + E \left[ \frac{1}{T} \left(\sum_{t=1}^{T} \eta_{t} \right) D_{\text{EBMP}}(\delta, \rho, R_T, \delta) \right] 
    \end{align}
    where inequality (1) comes from the application of the bounds in Equation~\ref{eq:lower_bound_init} and Equation~\ref{eq:lower_bound_next}, and equality (2) comes from linearity of expectations and the definition of $R_T$.
    
        Now, if we take the limit as $T \to \infty$ on both sides,  we get the following:
    \begin{align}
        &\lim_{T \to \infty} E \left[\frac{Z_{\pi_{\text{base}}, T}}{T} \right] \nonumber\\
        & \overset{(1)}{\geq} \lim_{T \to \infty} [ E\left[ \frac{\bar{m}}{T} D_{\text{EBMP}}(\xi, \rho, \bar{m}, \delta)\right] \nonumber \\  
        & \qquad - E[\frac{\bar{m}}{T} D_{\text{EBMP}}(\delta, \rho, R_T-\bar{m}, \delta)] \nonumber \\
        & \quad+ E \left[ \frac{1}{T} \left(\sum_{t=1}^{T} \eta_{t} \right) D_{\text{EBMP}}(\delta, \rho, R_T-\bar{m}, \delta) \right] ] \nonumber \\
        & \overset{(2)}{=} \lim_{T \to \infty} \left[ E\left[ \frac{\bar{m}}{T} D_{\text{EBMP}}(\xi, \rho, \bar{m}, \delta)\right] \right] \nonumber \\  
        & \qquad - \lim_{T \to \infty} \left[ E[\frac{\bar{m}}{T} D_{\text{EBMP}}(\delta, \rho, R_T-\bar{m}, \delta)]  \right]\nonumber \\
        & \quad+ \lim_{T \to \infty} \left[ E \left[ \frac{1}{T} \left(\sum_{t=1}^{T} \eta_{t} \right) D_{\text{EBMP}}(\delta, \rho, R_T-\bar{m}, \delta) \right] \right] \nonumber \\
        & \overset{(3)}{\geq} \liminf_{T \to \infty} \left[ E\left[ \frac{\bar{m}}{T} D_{\text{EBMP}}(\xi, \rho, \bar{m}, \delta)\right] \right] \nonumber \\  
        & \qquad - \limsup_{T \to \infty} \left[ E[\frac{\bar{m}}{T} D_{\text{EBMP}}(\delta, \rho, R_T-\bar{m}, \delta)]  \right]\nonumber \\
        & \quad+ \liminf_{T \to \infty} \left[ E \left[ \frac{1}{T} \left(\sum_{t=1}^{T} \eta_{t} \right) D_{\text{EBMP}}(\delta, \rho, R_T-\bar{m}, \delta) \right] \right] \nonumber \\
        & \overset{(4)}{\geq} E\left[ \liminf_{T \to \infty}\frac{\bar{m}}{T} D_{\text{EBMP}}(\xi, \rho, \bar{m}, \delta) \right] \nonumber \\
        &- E\left[ \limsup_{T \to \infty}\frac{\bar{m}}{T} D_{\text{EBMP}}(\delta, \rho, R_T-\bar{m}, \delta)\right]\nonumber \\
        &+  E \left[\liminf_{T \to \infty} \left(\sum_{t=1}^{T} \frac{\eta_{t}}{T} \right) D_{\text{EBMP}}(\delta, \rho, R_T-\bar{m}, \delta)  \right] \nonumber \\
        & \overset{(5)}{ = } E \left[\liminf_{T \to \infty} \left(\sum_{t=1}^{T} \frac{\eta_{t}}{T} \right) D_{\text{EBMP}}(\delta, \rho, R_T-\bar{m}, \delta)  \right] \nonumber \\
        %& \overset{(2)}{=} \lim_{T \to \infty} E \left[ \left(\sum_{t=1}^{T} \frac{\eta_{t}}{T} \right) D_{\text{EBMP}}(\delta, \rho, R_T, \delta) \right] \nonumber \\
        %& \overset{(3)}{=} \lim_{T \to \infty}\frac{1}{T} E[\bar{m}] E\left[ \lim_{T \to \infty} D_{\text{EBMP}}(\xi, \rho, \bar{m}, \delta) \right] \nonumber \\
        %&-\lim_{T \to \infty}\frac{1}{T} E[\bar{m}] E\left[ \lim_{T \to \infty} D_{\text{EBMP}}(\delta, \rho, R_T-\bar{m}, \delta) \right] \nonumber \\
        %& + E \left[\lim_{T \to \infty} \left(\sum_{t=1}^{T} \frac{\eta_{t}}{T} \right) D_{\text{EBMP}}(\delta, \rho, R_T-\bar{m}, \delta)  \right] \nonumber \\
        %& \overset{(4)}{=} 
        %&+ E[\eta] E\left[ \lim_{T \to \infty} D_{\text{EBMP}}(\delta, \rho, R_T-\bar{m}, \delta) \right] \nonumber \\
        & \overset{(6)}= E\left[ \liminf_{T \to \infty} \left(\sum_{t=1}^{T} \frac{\eta_{t}}{T} \right) \right] E\left[ \liminf_{T \to \infty} D_{\text{EBMP}}(\delta, \rho, R_T-\bar{m}, \delta) \right] \nonumber \\
        & \overset{(7)}= E[\eta] E\left[ \liminf_{T \to \infty} D_{\text{EBMP}}(\delta, \rho, R_T-\bar{m}, \delta) \right] \nonumber \\
        & \overset{(8)}{\geq} E[\eta] \cdot (\mathit{WD}(p_\delta, p_\rho) + E[d(\rho, \delta)]) \nonumber \\
        & = E[\eta] \cdot D_{\text{min}} \nonumber
    \end{align}
    where inequality (1) comes from the application of the limit to Equation~\ref{eq:lower_bound_avg_work}; equality (2) comes from the limit properties; inequality (3) comes from the definition of $\liminf$ and $\limsup$; inequality (4) comes from  the application of Fatou's lemma for the $\liminf$ and the Reverse Fatou's lemma for the $\limsup$; equality (5) comes from the fact that the first two terms go to zero as $T \to \infty$ since and $\bar{m}$ is upper bounded by a constant $m_{\text{cap}}$ and $D_{\text{EBMP}}$ is also upper bounded by a constant, more specifically the sum of $E[d(\rho, \delta)]$ and the generalized diameter of the euclidean region where the requests are being picked up; equality (6) comes from the fact that the expectation can be distributed since $\frac{1}{T} \left(\sum_{t=1}^{T} \eta_{t} \right)$ and $D_{\text{EBMP}}$ are independent; equality (7) comes from the application of the law of large numbers which, since $\eta_t$ are i.i.d, results in $\liminf_{T \to \infty} \frac{1}{T} \left(\sum_{t=1}^{T} \eta_{t} \right) = E[\eta]$; inequality (8) comes from the analytical results presented in \cite{Treleaven2013} where they show that $\liminf_{T \to \infty} D_{\text{EBMP}}(\delta, \rho, R_T, \delta) \geq \mathit{WD}(p_\delta, p_\rho) + E[d(\rho, \delta)]$, where $\mathit{WD}(p_\delta, p_\rho)$ corresponds to the Wasserstein distance required to transform $p_\delta$ into $p_\rho$ distribution. The expression $\mathit{WD}(p_\delta, p_\rho) + E[d(\rho, \delta)]$ is no longer a random variable and hence can be moved outside of the expectation. The last equality comes from the definition $D_{\text{min}} \triangleq \mathit{WD}(p_\delta, p_\rho) + E[d(\rho, \delta)]$. 
    
    Now, we can finally conclude that $E \left[\frac{Z_{\pi_{\text{base}}, T}}{T} \right]$ as $T \to \infty$ is lower bounded by $E[\eta] \cdot D_{\text{min}}$, and hence if $m < E[\eta] \cdot D_{\text{min}}$, the policy $\pi_{\text{base}}$ is asymptotically unstable since $m \cdot T < E \left[Z_{\pi_{\text{base}}, T} \right] $ as $T \to \infty$.
    
\end{proof}

\section{Numerical studies}
In this section we evaluate the performance of our algorithm using a real taxi data set for the city of San Francisco \cite{epfl-mobility-20090224}. We compare the performance of our algorithm against three benchmarks: a greedy policy, instantaneous assignment with reassignment (IA-RA), and a rollout-based algorithm over the entire map as proposed in~\cite{Garces_2023}. We provide a comparison of run-time of our two-phase approach and the rollout-based approach~\cite{Garces_2023} to empirically verify the reduction in run-time associated with our two-phase approach. We verify our theoretical results in the number of taxis in the fleet required for stability by executing our algorithm for larger time horizons and plotting the number of outstanding requests at each time step. We empirically verify that for $m$ chosen in the range given by Theorem \ref{theorem:standard_IA_stability}, and Theorem \ref{theorem:min_num_taxis}, our proposed approach is stable in the sense that the number of outstanding requests is uniformly bounded over time. 

\subsection{Experimental Setup}
Our numerical results consider a section of $1500m \times 1500m$ in San Francisco with $859$ nodes and $1959$ edges. For the comparison studies we consider a horizon length of $T=60$, while for the stability results we consider $T=180$. All experiments were executed in an AMD Threadripper PRO WRX80. All individual results correspond to an average over 20 different trials with different instantiations of the random variables.
%for the number of requests $\eta$, pickup locations $\rho,$, dropoff locations $\delta$, and the initial locations of the taxis $\xi$.

\subsection{Estimating probability distributions}
For our experiments, we estimate $\tilde{p}_{\eta}$, $\tilde{p}_{\rho|\eta}$, and $\tilde{p}_{\delta}$ using historical trip data from several taxis in San Francisco \cite{epfl-mobility-20090224}. We divide the historical data in 1-hour intervals, where each time step $t$ spans 1 minute. We empirically estimate $\tilde{p}_\eta$ by using the number of requests that arrive at each time step within each 1-hour time span. The distributions $\tilde{p}_\rho$ and $\tilde{p}_{\delta|\rho}$ are derived from the relative frequency of historical requests that originated and ended inside the map. %We estimate $\tilde{p}_{\delta|\rho}$ using relative frequency of pickup-dropoff pairs.

\subsection{Calculated values for theoretical results}\label{subsec:choice_m}
For our experiments, we consider $\tilde{p}_{\eta}$ for an hour in which $E[\eta] = 1$ (we get around $60$ requests per hour). For simplicity, we assume that $\xi$ is distributed according to the marginal probability distribution $p_{\delta}$, and hence we find that $E[d(\xi, \rho)] \approx 15$. We use $p_{l_{\text{rand}}}$ and $p_{\rho}$ to calculate $E[d(l_{\text{rand}}, \rho)] \approx 13$. We use $p_{\rho}$ and $p_{\delta|\rho}$ to calculate $E[d(\rho, \delta)] \approx 15$. From this we get that the sufficient number of taxis for stability of our two-phase approach is $m \geq \max(15,13) + 15 = 30$ from Theorem~\ref{theorem:standard_IA_stability}. We approximate the Wasserstein distance $\mathit{WD}(\delta, \rho)$ using the procedure suggested in \cite{spieser2014}. We obtain $\mathit{WD}(\delta, \rho) \approx 1.87$. From this, we get that asymptotically, the minimum number of taxis needed for stability as $T \to \infty$ is $m > 1.87 + 15$, rounding to next integer $m \geq 17$ (Theorem~\ref{theorem:min_num_taxis}).

\subsection{Implementation details for two-phase approach}
\label{subse:implementation}
We execute $2000$ Monte-Carlo simulations with certainty equivalence to approximate the expected cost associated with each potential control in the one-step lookahead step of the rollout for the local planner. We also consider a planning horizon $t_h=10$ for the rollout, and a capacity of $m_{\text{lim}}=10$ taxis per sector, based on the computational resources available. %We speculate that a larger $m_{\text{lim}}$ or $t_h$ will result in longer run-times, in exchange for slightly better performance.

\subsection{Benchmarks}
\label{subsec:benchmarks}
In this section, we discuss the details of the benchmarks to be used as comparisons for our performance results.

\textbf{Greedy policy:} Each taxi moves towards its closest request without coordinating with other taxis. This method does not consider future demand.

\textbf{Instantaneous assignment (IA-RA):} It solves a matching problem between available taxis and outstanding requests at every time step using an auction algorithm \cite{Bertsekas1979Auction, bertsekas1998network}. This method does not consider future demand.

\textbf{One-at-a-time rollout-based global routing:} performs rollout over the entire map using the procedure described in the scalability section of~\cite{Garces_2023}. We set the planning horizon to $t_h=10$ as suggested in the paper. We run the same number of MC simulations as with our approach. This method considers expected future demand.% and routes taxis accordingly.

\subsection{Performance results}
This section includes the results for the performance and the execution time of our two-phase approach. 
%First, we evaluate the performance of our proposed method compared to the benchmarks described in Sec.~\ref{subsec:benchmarks}.

% \begin{figure}[ht]
%         \centering
%         \includegraphics[width=0.48\textwidth]{images/cost_comparison_latest_.png}
%         \caption{\small{Performance of two-phase approach compared to all the other benchmarks. Results are averaged over 20 trials.}}
%         \label{fig:results_relative_performance}
% \end{figure}

As shown in Fig.~\ref{fig:results_relative_performance}, our method results in a comparable performance to the rollout-based global routing \cite{Garces_2023}. For lower number of taxis, when $m < 17$, our method is unstable.
After we surpass $17$ taxis, standard IA-RA starts being stable and performs similarly to the rollout-based global routing \cite{Garces_2023}. As shown in the graphs, for $m\geq 30$, our proposed method results in a lower cost than IA-RA, resulting in a $5\%$ to $18\%$ improvement, sometimes even outperforming the rollout-based global routing thanks to the smaller sampling space associated with each sector. Since both rollout-based methods are running the same number of MC simulations, a smaller sample space leads to better approximation of the expectation in Eq.(\ref{eq:Bellman_eq_ma}). 

To better understand the advantages of our proposed method, we compare the execution time of our proposed two-phase approach to the rollout-based global routing.
% \begin{figure}[ht]
%         \centering
%         \includegraphics[width=0.48\textwidth]{images/time_comparison_latest.png}
%         \caption{\small{Runtime of two-phase approach compared to rollout-based global routing method. Results are averaged over 10 trajectories.}}
%         \label{fig:results_relative_runtime}
% \end{figure}
\begin{figure}
         \centering
         \vspace{5pt}
         \includegraphics[width=0.75\linewidth]{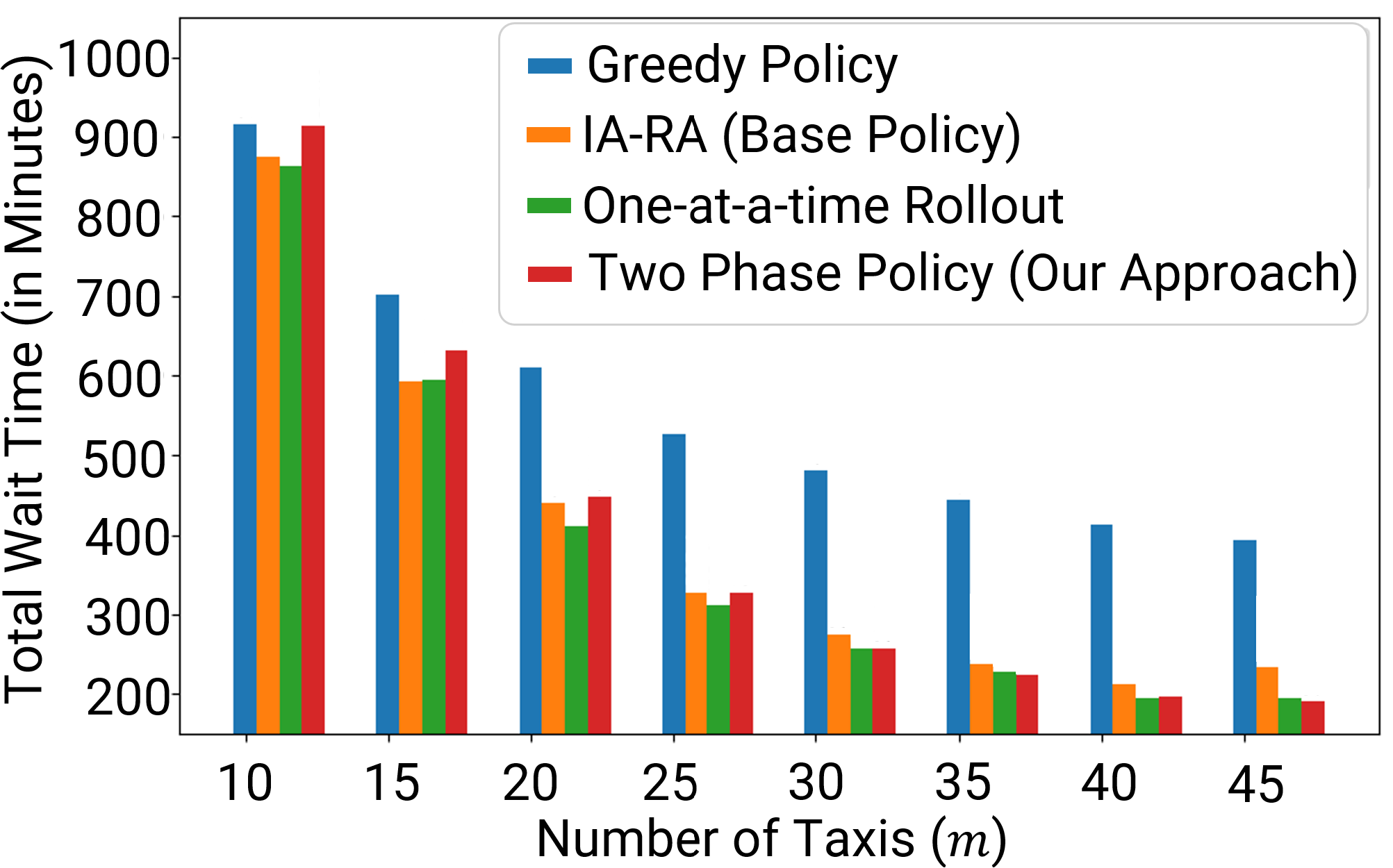}
         \vspace{-10pt}
         \caption{Total wait time over all requests of our two-phase approach and the benchmarks.}
         \label{fig:results_relative_performance}
\end{figure}
\begin{figure}
         \centering
         \includegraphics[width=0.85\linewidth]{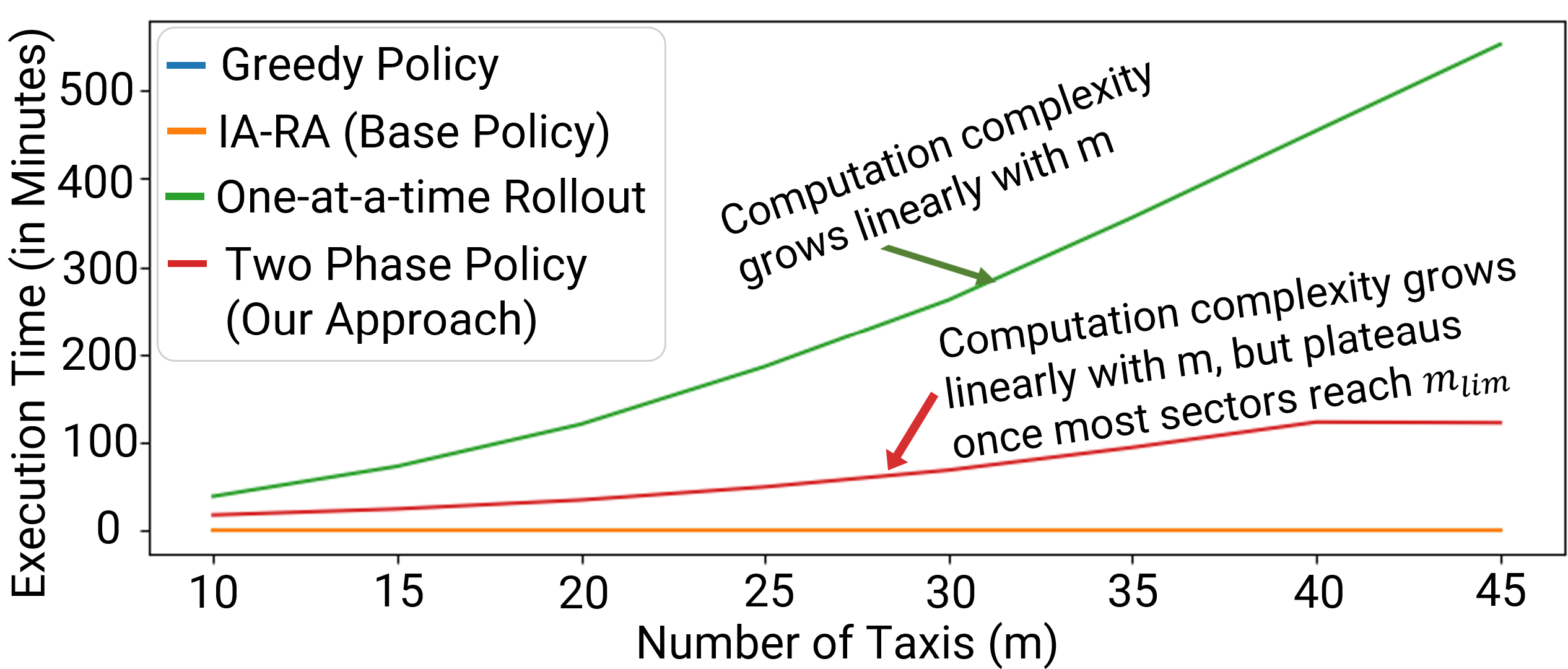}
         \vspace{-10pt}
         \caption{Execution time comparison between our two-phase approach and the benchmarks.}
         \label{fig:results_relative_runtime}
         \vspace{-8pt}
     \end{figure}
% \begin{figure}
%      \centering
%      \begin{subfigure}{\linewidth}
%          \centering
%          \vspace{10pt}
%          \includegraphics[width=0.85\linewidth, trim={0 0 0 0.55cm},clip]{images/cost_comparison_latest.png}\vspace{-10pt}
%          \caption{Performance (averaged over 20 trials) of two-phase approach compared to all the other benchmarks.}
%          \label{fig:results_relative_performance}
%      \end{subfigure}
%      \vfill
%      \begin{subfigure}{\linewidth}
%          \centering
%          \includegraphics[width=0.85\linewidth]{images/time_comparison2.png}\vspace{-10pt}
%          \caption{Execution time (averaged over 10 trials) comparison of policies, including two-phase approach and one-at-a-time rollout over the entire map.}
%          \label{fig:results_relative_runtime}
%          \vspace{-10pt}
%      \end{subfigure}
% \end{figure}
Fig.~\ref{fig:results_relative_runtime} shows our method results in significantly lower run-times than the rollout-based global routing. Execution time for our method still grows linearly with the number of taxis, but it eventually plateaus once the number of taxis in each sector reaches $m_{\text{lim}}$. This shows that partitioning the map and solving sub-problems in parallel results in a faster execution with similar wait time compared to one-at-a-time rollout over the entire map.

\subsection{Stability of two-phase approach}

Fig.~\ref{fig:stability_result} shows the stability results of our two-phase approach with various number of taxis over a horizon of 3 hours. Without enough taxis, $m< 17$, for which IA-RA is shown to be unstable, our approach shows an increasing number of outstanding requests over time. However, with sufficient numbers of taxis ($m=25,35$), we see that both the IA-RA policy and our two-phase approach has a bounded number of outstanding requests over a large horizon of $180$ minutes.

\begin{figure}
    \centering
    \vspace{-5pt}
    \includegraphics[width=0.78\linewidth]{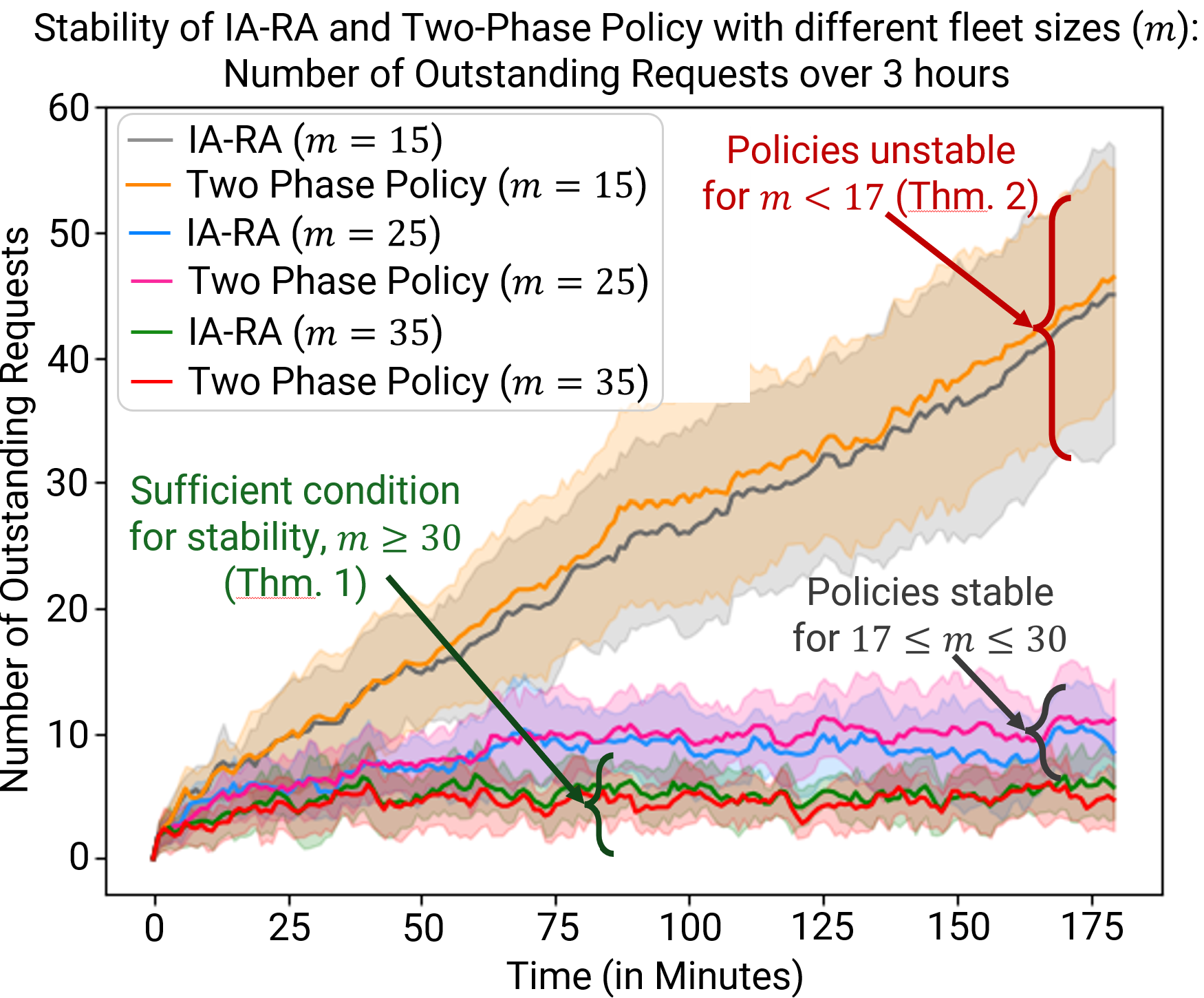}
    \vspace{-5pt}
    \caption{\small{Stability of IA-RA and two-phase policy in terms of the means (lines) and standard deviations (shaded regions) of the number of outstanding requests.}}
    \label{fig:stability_result}
\end{figure}

\section{Conclusion}
In this paper, we propose an approximation algorithm that allows us to apply one-at-a-time rollout to a large scale urban environment. We provide a necessary and a sufficient conditions for the total fleet size $m$ to make the instantaneous assignment base policy stable, which is key to guarantee rollout's convergence to a near-optimal policy. We also verify this results in simulation with a real dataset~\cite{epfl-mobility-20090224}. As future work, we plan on relaxing the assumption of unit travel time for the taxis in the fleet. Even though the algorithm would still work for this more realistic setting, we need to derive new theoretical results to take into account this change.

\section{Appendix}
In this appendix, we show that:
\begin{align}
    & \min_{l_{R_{t-1}+1}, \dots, l_{R_{t}} } E \left[ \sum_{q=R_{t-1}+1}^{R_{t}} d(l_{r_q}, \rho_{r_q}) + d(\rho_{r_q}, \delta_{r_q}) \right] \nonumber \\
    & \geq E \left[ \min_{l_{R_{t-1}+1}, \dots, l_{R_{t}} } \sum_{q=R_{t-1}+1}^{R_{t}} d(l_{r_q}, \rho_{r_q}) + d(\rho_{r_q}, \delta_{r_q}) \right] \nonumber
\end{align}

Let's first denote $Z_t(\vec{\ell})$ as the total distance required to service all requests in a given time step $t$ using a subset of the fleet $\vec{\ell}$, such that:

\begin{align*}
    Z_t(\vec{\ell}) = \sum_{i=1}^{\eta_{t}} d(l_i, \rho_i) + d(\rho_i, \delta_{i})
\end{align*}

Choosing a subset of the fleet $\vec{\ell}$ that minimizes this quantity, results in the following expression:

\begin{align*}
    Z_t(\vec{\ell}) \geq \min_{\vec{\ell}'} Z_t(\vec{\ell}')
\end{align*}

If we apply the expectation operator to both sides of the inequality, we get:

\begin{align*}
    E[Z_t(\vec{\ell})] \geq E[\min_{\vec{\ell'}} Z_t(\vec{\ell'})]
\end{align*}

Now if we take the minimum over the assignment of taxis to the subset of the fleet $\vec{\ell}$ on both sides of the inequality, we get:

\begin{align*}
    \min_{\vec{\ell^{''}}}E[Z_t(\vec{\ell^{''}})] & \geq \min_{\vec{\ell^{''}}} E[\min_{\vec{\ell'}} Z_t(\vec{\ell'})] \\
    & \overset{(1)}{=} E[\min_{\vec{\ell'}} Z_t(\vec{\ell'})]
\end{align*}

Where equality (1) follows from the fact that $E[\min_{\vec{\ell'}} Z_t(\vec{\ell'})]$ does not depend on $\vec{\ell^{''}}$. From this result, we get that $\min_{\vec{\ell^{''}}}E[Z_t(\vec{\ell^{''}})] \geq E[\min_{\vec{\ell'}} Z_t(\vec{\ell'})]$, and we can conclude that the original inequality 

\begin{align}
    & \min_{l_{R_{t-1}+1}, \dots, l_{R_{t}} } E \left[ \sum_{q=R_{t-1}+1}^{R_{t}} d(l_{r_q}, \rho_{r_q}) + d(\rho_{r_q}, \delta_{r_q}) \right] \nonumber \\
    & \geq E \left[ \min_{l_{R_{t-1}+1}, \dots, l_{R_{t}} } \sum_{q=R_{t-1}+1}^{R_{t}} d(l_{r_q}, \rho_{r_q}) + d(\rho_{r_q}, \delta_{r_q}) \right] \nonumber
\end{align}
holds.

\label{sec:appendix}

\bibliographystyle{IEEEtran}
\bibliography{IEEEabrv, main}

% Generated by IEEEtran.bst, version: 1.14 (2015/08/26)
\begin{thebibliography}{10}
\providecommand{\url}[1]{#1}
\csname url@samestyle\endcsname
\providecommand{\newblock}{\relax}
\providecommand{\bibinfo}[2]{#2}
\providecommand{\BIBentrySTDinterwordspacing}{\spaceskip=0pt\relax}
\providecommand{\BIBentryALTinterwordstretchfactor}{4}
\providecommand{\BIBentryALTinterwordspacing}{\spaceskip=\fontdimen2\font plus
\BIBentryALTinterwordstretchfactor\fontdimen3\font minus
  \fontdimen4\font\relax}
\providecommand{\BIBforeignlanguage}[2]{{%
\expandafter\ifx\csname l@#1\endcsname\relax
\typeout{** WARNING: IEEEtran.bst: No hyphenation pattern has been}%
\typeout{** loaded for the language `#1'. Using the pattern for}%
\typeout{** the default language instead.}%
\else
\language=\csname l@#1\endcsname
\fi
#2}}
\providecommand{\BIBdecl}{\relax}
\BIBdecl

\bibitem{Bidarian2023}
\BIBentryALTinterwordspacing
N.~Bidarian, ``Regulators give green light to driverless taxis in san
  francisco,'' \emph{CNN}, 2023. [Online]. Available:
  \url{https://www.cnn.com/2023/08/11/tech/robotaxi-vote-san-francisco/index.html}
\BIBentrySTDinterwordspacing

\bibitem{Muller2023}
\BIBentryALTinterwordspacing
J.~Muller, ``Robotaxis hit the accelerator in growing list of cities
  nationwide,'' \emph{Axios}, 2023. [Online]. Available:
  \url{https://www.axios.com/2023/08/29/cities-testing-self-driving-driverless-taxis-robotaxi-waymo}
\BIBentrySTDinterwordspacing

\bibitem{Kondor2022}
D.~Kondor, I.~Bojic, G.~Resta, F.~Duarte, P.~Santi, and C.~Ratti, ``The cost of
  non-coordination in urban on-demand mobility,'' \emph{Scientific Reports},
  vol.~12, 03 2022.

\bibitem{BERBEGLIA20108}
\BIBentryALTinterwordspacing
G.~Berbeglia, J.-F. Cordeau, and G.~Laporte, ``Dynamic pickup and delivery
  problems,'' \emph{European Journal of Operational Research}, vol. 202, no.~1,
  pp. 8--15, 2010. [Online]. Available:
  \url{https://www.sciencedirect.com/science/article/pii/S0377221709002999}
\BIBentrySTDinterwordspacing

\bibitem{Duan2014}
\BIBentryALTinterwordspacing
R.~Duan and S.~Pettie, ``Linear-time approximation for maximum weight
  matching,'' \emph{J. ACM}, vol.~61, no.~1, 2014. [Online]. Available:
  \url{https://doi.org/10.1145/2529989}
\BIBentrySTDinterwordspacing

\bibitem{Bertsimas2019OnlineVR}
D.~Bertsimas, P.~Jaillet, and S.~Martin, ``Online vehicle routing: The edge of
  optimization in large-scale applications,'' \emph{Oper. Res.}, vol.~67, pp.
  143--162, 2019.

\bibitem{alonso2017predictive}
J.~Alonso-Mora, A.~Wallar, and D.~Rus, ``Predictive routing for autonomous
  mobility-on-demand systems with ride-sharing,'' in \emph{2017 IEEE/RSJ
  International Conference on Intelligent Robots and Systems (IROS)}.\hskip 1em
  plus 0.5em minus 0.4em\relax IEEE, 2017, pp. 3583--3590.

\bibitem{LOWALEKAR201871}
\BIBentryALTinterwordspacing
M.~Lowalekar, P.~Varakantham, and P.~Jaillet, ``Online spatio-temporal matching
  in stochastic and dynamic domains,'' \emph{Artificial Intelligence}, vol.
  261, pp. 71--112, 2018. [Online]. Available:
  \url{https://www.sciencedirect.com/science/article/pii/S0004370218302030}
\BIBentrySTDinterwordspacing

\bibitem{iglesias2018data}
R.~Iglesias, F.~Rossi, K.~Wang, D.~Hallac, J.~Leskovec, and M.~Pavone,
  ``Data-driven model predictive control of autonomous mobility-on-demand
  systems,'' in \emph{2018 IEEE international conference on robotics and
  automation (ICRA)}.\hskip 1em plus 0.5em minus 0.4em\relax IEEE, 2018, pp.
  6019--6025.

\bibitem{gammelli2022graph}
D.~Gammelli, K.~Yang, J.~Harrison, F.~Rodrigues, F.~Pereira, and M.~Pavone,
  ``Graph meta-reinforcement learning for transferable autonomous
  mobility-on-demand,'' in \emph{Proceedings of the 28th ACM SIGKDD Conference
  on Knowledge Discovery and Data Mining}, 2022, pp. 2913--2923.

\bibitem{enders2023hybrid}
T.~Enders, J.~Harrison, M.~Pavone, and M.~Schiffer, ``Hybrid multi-agent deep
  reinforcement learning for autonomous mobility on demand systems,'' in
  \emph{Learning for Dynamics and Control Conference}.\hskip 1em plus 0.5em
  minus 0.4em\relax PMLR, 2023, pp. 1284--1296.

\bibitem{Garces_2023}
D.~Garces, S.~Bhattacharya, S.~Gil, and D.~Bertsekas, ``Multiagent
  reinforcement learning for autonomous routing and pickup problem with
  adaptation to variable demand,'' in \emph{2023 {IEEE} International
  Conference on Robotics and Automation ({ICRA})}.\hskip 1em plus 0.5em minus
  0.4em\relax {IEEE}, may 2023.

\bibitem{Bertsekas2021PI}
D.~Bertsekas, ``Multiagent reinforcement learning: Rollout and policy
  iteration,'' \emph{IEEE/CAA Journal of Automatica Sinica}, vol.~8, no.~2, pp.
  249--272, 2021.

\bibitem{bertsekas2020rollout}
\BIBentryALTinterwordspacing
------, \emph{Rollout, Policy Iteration, and Distributed Reinforcement
  Learning}, ser. Athena scientific optimization and computation series.\hskip
  1em plus 0.5em minus 0.4em\relax Athena Scientific., 2020. [Online].
  Available: \url{https://books.google.com/books?id=Hbo-EAAAQBAJ}
\BIBentrySTDinterwordspacing

\bibitem{Bertsekas2022AlphaZero}
------, \emph{Lessons from AlphaZero for Optimal, Model Predictive, and
  Adaptive Control}.\hskip 1em plus 0.5em minus 0.4em\relax Nashua, NH, USA:
  Athena Scientific, 2022.

\bibitem{gerkey2004formal}
B.~P. Gerkey and M.~J. Matari{\'c}, ``A formal analysis and taxonomy of task
  allocation in multi-robot systems,'' \emph{The International journal of
  robotics research}, vol.~23, no.~9, pp. 939--954, 2004.

\bibitem{Zhang2016Queue}
\BIBentryALTinterwordspacing
R.~Zhang and M.~Pavone, ``Control of robotic mobility-on-demand systems: A
  queueing-theoretical perspective,'' \emph{The International Journal of
  Robotics Research}, vol.~35, no. 1–3, p. 186–203, jan 2016. [Online].
  Available: \url{https://doi.org/10.1177/0278364915581863}
\BIBentrySTDinterwordspacing

\bibitem{Vazifeh2018}
M.~Vazifeh, P.~Santi, G.~Resta, S.~Strogatz, and C.~Ratti, ``Addressing the
  minimum fleet problem in on-demand urban mobility,'' \emph{Nature}, vol. 557,
  05 2018.

\bibitem{Treleaven2013}
K.~Treleaven, M.~Pavone, and E.~Frazzoli, ``Asymptotically optimal algorithms
  for one-to-one pickup and delivery problems with applications to
  transportation systems,'' \emph{IEEE Transactions on Automatic Control},
  vol.~58, no.~9, pp. 2261--2276, 2013.

\bibitem{spieser2014}
K.~Spieser, K.~Treleaven, R.~Zhang, E.~Frazzoli, D.~Morton, and M.~Pavone,
  ``Toward a systematic approach to the design and evaluation of automated
  mobility-on-demand systems: A case study in singapore,'' \emph{Road Vehicle
  Automation. Lecture Notes on Mobility}, pp. 229--245, 04 2014.

\bibitem{WU2006}
\BIBentryALTinterwordspacing
L.-Y. Wu, X.-S. Zhang, and J.-L. Zhang, ``Capacitated facility location problem
  with general setup cost,'' \emph{Computers \& Operations Research}, vol.~33,
  no.~5, pp. 1226--1241, 2006. [Online]. Available:
  \url{https://www.sciencedirect.com/science/article/pii/S0305054804002357}
\BIBentrySTDinterwordspacing

\bibitem{Bertsekas2020Auction}
\BIBentryALTinterwordspacing
D.~Bertsekas, ``Constrained multiagent rollout and multidimensional assignment
  with the auction algorithm,'' 2020. [Online]. Available:
  \url{https://arxiv.org/abs/2002.07407}
\BIBentrySTDinterwordspacing

\bibitem{Crouse2016}
D.~F. Crouse, ``On implementing 2d rectangular assignment algorithms,''
  \emph{IEEE Transactions on Aerospace and Electronic Systems}, vol.~52, no.~4,
  pp. 1679--1696, 2016.

\bibitem{ruschendorf1985wasserstein}
L.~R{\"u}schendorf, ``The wasserstein distance and approximation theorems,''
  \emph{Probability Theory and Related Fields}, vol.~70, no.~1, pp. 117--129,
  1985.

\bibitem{epfl-mobility-20090224}
M.~Piorkowski, N.~Sarafijanovic-Djukic, and M.~Grossglauser, ``{CRAWDAD}
  dataset epfl/mobility (v. 2009-02-24),'' Downloaded from
  \url{https://crawdad.org/epfl/mobility/20090224}, Feb. 2009.

\bibitem{Bertsekas1979Auction}
D.~Bertsekas, ``A distributed algorithm for the assignment problem,''
  \emph{Lab. for Information and Decision Systems Report}, 05 1979.

\bibitem{bertsekas1998network}
\BIBentryALTinterwordspacing
------, \emph{Network Optimization: Continuous and Discrete Models}, ser.
  Athena scientific optimization and computation series.\hskip 1em plus 0.5em
  minus 0.4em\relax Athena Scientific, 1998. [Online]. Available:
  \url{https://books.google.com/books?id=qUUxEAAAQBAJ}
\BIBentrySTDinterwordspacing

\end{thebibliography}

\end{document}